\documentclass[11pt]{article}
\usepackage{hdb_macros}
\usepackage[margin=1in]{geometry}

\newcommand{\Lfinal}{\lat_{\textnormal{final}}}
\newcommand{\Lint}{\lat_{\textnormal{int}}}

\newcommand{\cS}{\mathcal{S}}
\newcommand{\cG}{\mathcal{G}}
\newcommand{\cA}{\mathcal{A}}
\newcommand{\Bfinal}{B_{\textnormal{final}}}
\newcommand{\Gfinal}{G_{\textnormal{final}}}
\newcommand{\Cfinal}{\C_{\textnormal{final}}}

\newcommand{\mindistp}{\lambda_1^{(p)}}

\newcommand{\Cint}{\C_{\textnormal{int}}}
\newcommand{\Crand}{\C_{\textnormal{rand}}}
\newcommand{\BCH}{\C_{\textnormal{BCH}}}
\newcommand{\GBCH}{G_{\textnormal{BCH}}}
\newcommand{\BBCH}{B_{\textnormal{BCH}}}
\newcommand{\BNCP}{B_{\textnormal{NCP}}}
\newcommand{\latBCH}{\lat_{\textnormal{BCH}}}
\newcommand{\latNCP}{\lat_{\textnormal{NCP}}}
\newcommand{\latRS}{\lat_{\textnormal{RS}}}
\newcommand{\latRSparam}{\lat_{\textnormal{RS}_{q,\ell}}}
\newcommand{\NCP}{\textnormal{NCP}}
\newcommand{\MDP}{\textnormal{MDP}}
\renewcommand{\C}{\mathcal{C}}
\newcommand{\Bint}{B_{\textnormal{int}}}
\newcommand{\Gint}{G_{\textnormal{int}}}

\newcommand{\RSparam}{\textnormal{RS}_{q,\ell}}

\newcommand{\W}{\cc{W}}
\newcommand{\FPT}{\cc{FPT}}

\newcommand{\fullornot}[2]{#1}

\usepackage[draft,multiuser,inline,nomargin]{fixme}
\fxusetheme{color}
\FXRegisterAuthor{h}{eh}{\color{blue}Huck}
\FXRegisterAuthor{j}{ej}{\color{red}Jo\~ao}
\FXRegisterAuthor{v}{ev}{\color{cyan}Venkat}
\date{\today}

\renewcommand{\le}{\leqslant}
\renewcommand{\leq}{\leqslant}

\renewcommand{\geq}{\geqslant}

\title{Parameterized Inapproximability of the\\ Minimum Distance Problem over all Fields and the\\ Shortest Vector Problem in all $\ell_p$ Norms\thanks{A preliminary version of this work appeared at STOC 2023~\cite{BCGR23}.}}

\author{Huck Bennett\thanks{University of Colorado Boulder. \email{huck.bennett@oregonstate.edu}. This work was mainly done while the author was at Oregon State University.}
\and Mahdi Cheraghchi\thanks{University of Michigan, Ann Arbor. \email{mahdich@umich.edu}.}
\and Venkatesan Guruswami\thanks{University of California, Berkeley. \email{venkatg@berkeley.edu}}
\and Jo\~{a}o Ribeiro\thanks{NOVA LINCS and NOVA School of Science and Technology -- Universidade Nova de Lisboa. \email{joao.ribeiro@fct.unl.pt}. This work was mainly done while the author was at Carnegie Mellon University.}} 
\date{}

\begin{document}

\maketitle
\thispagestyle{empty}
\listoffixmes
\begin{abstract}
We prove that the Minimum Distance Problem ($\MDP$) on linear codes over any fixed finite field and parameterized by the input distance bound is $\W[1]$-hard to approximate within any constant factor.
We also prove analogous results for the parameterized Shortest Vector Problem ($\SVP$) on integer lattices. Specifically, we prove that $\SVP$ in the $\ell_p$ norm is $\W[1]$-hard to approximate within any constant factor for any fixed $p >1$ and $\W[1]$-hard to approximate within a factor approaching $2$ for $p=1$.
(We show hardness under randomized reductions in each case.)

These results answer the main questions left open (and explicitly posed) by
Bhattacharyya, Bonnet, Egri, Ghoshal, Karthik C.\ S., Lin, Manurangsi, and Marx (Journal of the ACM, 2021) on the complexity of parameterized $\MDP$ and $\SVP$. 
For $\MDP$, they established similar hardness for \emph{binary} linear codes and left the case of general fields open. For $\SVP$ in $\ell_p$ norms with $p > 1$, they showed inapproximability within \emph{some} constant factor (depending on $p$) and left open showing such hardness for arbitrary constant factors. 
They also left open showing $\W[1]$-hardness even of exact SVP in the $\ell_1$ norm.
\end{abstract}
\newpage
\tableofcontents
\pagenumbering{roman}
\newpage
\pagenumbering{arabic}

\section{Introduction}
\label{sec:intro}

Error correcting codes and point lattices are fundamental mathematical objects, and computational problems on them have a wide range of applications in computer science including to robust communication, cryptography, optimization, complexity theory, and more. Indeed, because computational problems on codes and lattices are so ubiquitous, a highly active line of research spanning decades has worked to understand the complexity of the problems themselves.
In particular, a great deal of work has studied the complexity of the Minimum Distance Problem ($\MDP$) (and its affine version, the Nearest Codeword Problem ($\NCP$)) on linear error correcting codes~\cite{BMvT78,journals/jcss/AroraBSS97,Var97,journals/tit/DumerMS03,CW12}.
Similarly, a large amount of work has studied the complexity of the analogous problems on lattices, the Shortest Vector Problem ($\SVP$) (and its affine version, the Closest Vector Problem ($\CVP$))~\cite{vanEmdeBoas81,journals/jcss/AroraBSS97,DBLP:conf/stoc/Ajtai98,journals/siamcomp/Micciancio00,Kho05,HR12}.

In $\MDP_q$, the goal is, given a linear error correcting code $\C$ over a finite field $\F_q$ and a distance bound $k$ as input, to determine whether or not the minimum Hamming weight of a nonzero codeword in $\C$ is at most $k$.
Similarly, in $\SVP_p$ the goal is, given a lattice $\lat$ and a distance bound $k$ as input, to determine whether or not the $\ell_p$ norm of some nonzero vector in $\lat$ is at most $k$.\footnote{The $\ell_p$ norm used is fixed and independent of the input. One may also consider $\SVP$ with respect to arbitrary norms, but it is most commonly considered with respect to $\ell_p$ norms (and especially with respect to the Euclidean norm $\ell_2$) as is the case in this work.} 
One may also consider $\gamma$-approximate versions of these problems for $\gamma \geq 1$, which we denote by $\gamma$-$\MDP_q$ and $\gamma$-$\SVP_p$, respectively.
(In what follows we refer to linear error correcting codes over finite fields simply as ``codes.'' We define codes, lattices, and computational problems on them formally in~\cref{sec:coding-problems,sec:lattice-problems}.)

In the 1990s, the field of \emph{parameterized complexity}, in which the running time of an algorithm for a given computational problem is considered not just as a function of the problem's input size $n$ but also in terms of some parameter $k$, developed and matured.
The fundamental notion of efficiency in the study of parameterized algorithms is \emph{fixed-parameter tractability}, which means that the algorithm runs in time $f(k) \cdot \poly(n)$ for some (possibly fast-growing) function $f(k)$ depending on the parameter $k$ but not the input length.
A computational problem (formally, problem-parameter pair) with such an algorithm is called \emph{fixed-parameter tractable} (FPT), and the set of all such problems forms the complexity class $\FPT$.
On the other hand, the canonical notion of inefficiency for parameterized problems is $\W[1]$-hardness, which is analogous to $\NP$-hardness in the non-parameterized setting. 
To show $\W[1]$-hardness of a given problem, it suffices to give an \emph{FPT reduction} from a known $\W[1]$-hard problem to that problem.
Giving such a reduction in particular implies that the problem reduced to is not in $\FPT$ unless $\W[1] = \FPT$, which is widely conjectured not to be the case.
(Determining whether $\FPT$ is equal to $\W[1]$ is a major open question, and is the analog of the $\P$ versus $\NP$ question in the parameterized world.)
See the books by Downey and Fellows~\cite{DF99,Downey2013-di} for comprehensive references on parameterized complexity.

\medskip\noindent\textbf{Parameterized complexity of $\MDP$ and $\SVP$.} As part of the development of parameterized complexity as a whole, substantial interest arose in the parameterized complexity (specifically, $\W[1]$-hardness) of computational problems on codes and lattices. This was especially true for $\MDP$ and $\SVP$, where in each case the parameter $k$ of interest is the input distance bound.\footnote{In the parameterized setting, we consider $\SVP$ only on \emph{integer} lattices; otherwise the distance bound is not meaningful.}
Indeed, until recently, one of the major unresolved questions in parameterized complexity theory was to determine whether the Minimum Distance Problem on binary codes was $\W[1]$-hard. It was one of the few remaining open problems from~\cite{DF99}, and Downey and Fellows called it one of the ``most infamous'' such open problems in their follow-up book~\cite{Downey2013-di}.\footnote{\label{footnote:even-set}More precisely,~\cite{DF99,Downey2013-di} asked about the complexity of the Even Set Problem, which is equivalent to the dual formulation of the Minimum Distance Problem over $\F_2$.}
Similarly, the fixed-parameter (in)tractability of the Shortest Vector Problem in the $\ell_2$ norm was mentioned as an important unresolved question in~\cite{DF99,Downey2013-di}.
Interestingly, it is not known whether $\gamma$-$\MDP_q$ nor $\gamma$-$\SVP_p$ are in W[1].

In recent seminal work, Bhattacharyya, Bonnet, Egri, Ghoshal, Karthik C.\ S., Lin, Manurangsi, and Marx~\cite{BBEGSLMM21},  building on work of Lin~\cite{journals/jacm/Lin18}, resolved both of these questions in the affirmative. They in fact even showed that both parameterized $\MDP$ and $\SVP$ are hard to approximate. Specifically, they showed that for any constant $\gamma \geq 1$, $\gamma$-$\MDP_{2}$ is $\W[1]$-hard to approximate under randomized reductions, and that for any $p > 1$ and constant $\gamma$ satisfying $1 \leq \gamma < (1/2+1/2^p)^{-1/p}$, $\gamma$-$\SVP_p$ is $\W[1]$-hard to approximate under randomized reductions.

However, despite its major achievements,~\cite{BBEGSLMM21} still fell short of providing a complete understanding of the parameterized hardness of approximate $\MDP$ and $\SVP$.
To that end, they gave several open questions.
Specifically, the authors asked whether it was possible to show $\W[1]$-hardness of $\MDP$ over \emph{all} finite fields $\F_q$ (and not just for the binary case of $\F_2$).
They also asked about showing $\W[1]$-hardness of $\SVP$ in \emph{all} $\ell_p$ norms (specifically, they asked about $\gamma$-$\SVP_1$, for which they did not show hardness even in the exact case of $\gamma = 1$), and about showing $\W[1]$-hardness of $\gamma$-$\SVP_p$ with arbitrarily large constant $\gamma$ for (some) $p$ (they did not obtain such a result for any $p$).%
\footnote{In fact,~\cite{BBEGSLMM21} asked about such a result for $p \neq 2$ and claimed such a result in passing for $p = 2$. However, the claim was referring to a result from prior work (specifically, \cite{BGKM18}) that showed hardness only under stronger hypotheses. See \cref{remark:svp2-tensoring}.}
The first two of these three questions from~\cite{BBEGSLMM21}
were also asked as Open Questions $2$ and $3$, respectively, in a recent survey on approximation in parameterized complexity by Feldmann, Karthik C.\ S., Lee, and Manurangsi~\cite{FKLM20}, which discussed important open problems in the field as a whole.

\subsection{Our contributions}

In this work, we answer all three of the open questions of~\cite{BBEGSLMM21} discussed above, 
and provide a nearly complete picture of the parameterized inapproximability of $\MDP$ and $\SVP$. In each of the three theorems below (i.e., \cref{thm:hardness-mdp,thm:hardness-l1,thm:hardness-lp-tensor}) the parameter of interest is the input distance bound $k$.

We first give our hardness result for $\MDP$, which resolves the first open question from~\cite{BBEGSLMM21} (also asked as~\cite[Open Question 2]{FKLM20}).

\begin{restatable}{theorem}{hardnessmdp}
\label{thm:hardness-mdp}
For any fixed prime power $q$ and constant $\gamma \geq 1$, $\gamma$-$\MDP_q$ is $\W[1]$-hard under randomized FPT reductions with two-sided error.
\end{restatable}

Second, we settle the second open question from~\cite{BBEGSLMM21} (also asked as~\cite[Open Question 3]{FKLM20}) by showing the following hardness result for parameterized $\gamma$-$\SVP_p$ for all (finite) $p \geq 1$ and some $\gamma = \gamma(p)$.\footnote{We do not consider the case of $p = \infty$ because, as~\cite{BBEGSLMM21} notes, $\SVP$ in the $\ell_{\infty}$ norm is $\NP$-hard even when $k = 1$.} Indeed, in particular applies to the $\ell_1$ norm.
 It also shows hardness of approximation for larger factors $\gamma(p)$ for $p > 1$ than~\cite{BBEGSLMM21} does.
 
\begin{restatable}{theorem}{hardnesslone}\label{thm:hardness-l1}
For any fixed $p \in [1, \infty)$ and constant $\gamma\in[1,2^{1/p})$, $\gamma$-$\SVP_p$ is $\W[1]$-hard under randomized FPT reductions with two-sided error.
\end{restatable}

Finally, we establish the parameterized inapproximability of $\SVP$ with an arbitrary constant approximation factor in the $\ell_p$ norm for all fixed $p>1$. This resolves the third question from~\cite{BBEGSLMM21} mentioned above.

\begin{restatable}{theorem}{hardnesslp} \label{thm:hardness-lp-tensor}
For any fixed $p \in (1, \infty)$ and constant $\gamma \geq 1$,
$\gamma$-$\SVP_p$ is $\W[1]$-hard under randomized FPT reductions with two-sided error.
\end{restatable}

\begin{remark} \label{remark:svp2-tensoring}
We note that~\cite{BBEGSLMM21} erroneously claimed in a passing remark that the important Euclidean (i.e., $p=2$) special case of \cref{thm:hardness-lp-tensor} was already known.
However, that remark was in fact referring to a result from an earlier version of~\cite{BBEGSLMM21} (i.e.,~\cite{BGKM18}) that shows parameterized hardness of $\gamma$-$\SVP_2$ for arbitrary constant $\gamma \geq 1$, \emph{but only under the (randomized) Gap Exponential Time Hypothesis (Gap-ETH) or the Parameterized Inapproximability Hypothesis (PIH)}, which are stronger assumptions than $\W[1] \neq \cc{FPT}$.\footnote{Randomized Gap-ETH~\cite{Din16,MR17} states that there exist constants $\eps,c>0$ such that no randomized algorithm which is given as input a
$3$-CNF formula $F$ with $m$ clauses and runs in time $O(2^{cm})$ can distinguish with probability at least 2/3 between the cases where $F$ is satisfiable and where only at most a $(1 -\eps)$-fraction of clauses in $F$ are satisfiable.\\
PIH~\cite{LRSZ20} states that there exists a constant $\eps>0$ such that it is $\W[1]$-hard to approximate the Multicolored Densest Subgraph problem to within a $\gamma=1+\eps$ approximation factor. This corresponds to the problem where we are given as input a graph $G=(V,E)$ with the vertex set partitioned into $k$ sets $V_1,\dots,V_k$, and the goal is to select vertices $v_1\in V_1,\dots,v_k\in V_k$ that induce as many edges as possible in $G$.}
In particular, the result in \cref{thm:hardness-lp-tensor} was previously unknown for any $p$.
We thank Pasin Manurangsi~\cite{Man22} for clarifying this for us.
\end{remark}

We provide a technical overview of our arguments in \cref{sec:tech-overview} and provide formal proofs of \cref{thm:hardness-mdp,thm:hardness-l1,thm:hardness-lp-tensor} in \cref{sec:NCPtoMDP,sec:CVPtoSVP,sec:NCPtoSVPtensor}, respectively.

\medskip\noindent\textbf{Fine-grained hardness of parameterized $\MDP$ and $\SVP$.}
Our reductions also directly yield improved results concerning the fine-grained hardness of $\gamma$-$\MDP_q$ and $\gamma$-$\SVP_p$ under Gap-ETH.
Leveraging results from~\cite{BGKM18,BBEGSLMM21}, Manurangsi~\cite{Man20} showed that there are no (possibly randomized) algorithms running in time $f(k)\cdot n^{o(k)}$ for $\gamma$-$\NCP_q$ (respectively, time $f(k)\cdot n^{o(k^p)}$ for $\gamma$-$\CVP_p$) with any function $f$, $\gamma\geq 1$, and prime power $q$ (respectively, $p\geq 1$), where $n$ is the dimension of the input code (respectively, the rank of the input lattice) and $k$ is the input distance bound (in each case) assuming randomized Gap-ETH.

By inspection, our FPT reductions from approximate $\NCP_q$ to approximate $\MDP_q$ and from approximate $\CVP_p$ to approximate $\SVP_p$ in \cref{sec:NCPtoMDP,sec:CVPtoSVP}, respectively, transform the distance parameter $k$ into $k'=O(k)$ (for formal statements, see \cref{thm:redNCPtoMDP,thm:redCVPtoSVP}).
Therefore, we can combine these reductions with the results from~\cite{Man20} to immediately obtain the following results on the parameterized fine-grained hardness of $\MDP$ and $\SVP$, which imply that the brute force solution to each of these problems is essentially optimal under randomized Gap-ETH.%
\footnote{We note that we work with the standard (in the non-parameterized setting) formulation of $\gamma$-$\SVP_p$ throughout the paper, where the goal is to decide whether the input lattice has a nonzero vector $\vec{x}$ with $\norm{\vec{x}}_p \leq k$ or if all such vectors have $\ell_p$ norm greater than $\gamma k$. On the other hand,~\cite{BBEGSLMM21,Man20} work with an equivalent but different parameterization of the problem, which asks whether the input lattice has a nonzero vector $\vec{x}$ with $\norm{\vec{x}}_p^p \leq k$ or if the $p$th power of the $\ell_p$ norm of all such vectors is greater than $\gamma k$. This discrepancy leads to certain runtimes and approximation factors in our work being off by a power of $p$ from~\cite{BBEGSLMM21,Man20}.}

\begin{restatable}{theorem}{finegrainedMDP}\label{thm:finegrainedMDP}
For any function $f$, fixed prime power $q$, and every $\gamma\in\left[1,\frac{2q}{2q-1}\right)$ it holds that, assuming randomized Gap-ETH, there is no randomized algorithm running in time $f(k)\cdot n^{o(k)}$ for deciding $\gamma$-$\MDP_q$, where $n$ is the dimension of the input code  and $k$ is the input distance bound.
\end{restatable}
\begin{restatable}{theorem}{finegrainedSVP}\label{thm:finegrainedSVP}
For any function $f$, fixed real number $p\geq 1$, and every $\gamma\in[1,2^{1/p})$ it holds that, assuming randomized Gap-ETH, there is no randomized algorithm running in time $f(k)\cdot n^{o(k^p)}$ for deciding $\gamma$-$\SVP_p$, where $n$ is the rank of the input lattice and $k$ is the input distance bound.
\end{restatable}

Previously, \cref{thm:finegrainedMDP} was only known to hold for $q=2$, and \cref{thm:finegrainedSVP} was only known to hold for $p>1$ and with approximation factors $\gamma = \gamma(p) < (1/2+1/2^p)^{-1/p}$ that are smaller than those achieved by \cref{thm:finegrainedSVP}; see~\cite{BGKM18,Man20}.

Interestingly, the standard technique of tensoring instances of $\MDP$ or $\SVP$ to boost the approximation factor cannot be used to prove fine-grained hardness results as above, because the distance parameter $k$ is mapped to $k'=k^c$ for $c>1$.
This motivates the search for FPT reductions that preserve the parameter $k$ up to a linear factor (i.e., for which $k'=O(k)$) while simultaneously showing hardness for as large an approximation factor $\gamma$ as possible.
Our pre-tensoring hardness reductions for $\MDP_q$ and $\SVP_p$ in \cref{sec:NCPtoMDP,sec:CVPtoSVP} also yield $k'=O(k)$, and it would be interesting to come up with improved reductions that achieve larger approximation factors. 
Moreover, we note that although we obtain better $\W[1]$-hardness of approximation for $\SVP_p$ with $p >1$ from the reduction in \cref{sec:NCPtoSVPtensor}, we in fact get better fine-grained hardness from our reduction in \cref{sec:CVPtoSVP}. (The reduction in \cref{sec:CVPtoSVP} also has the advantage of showing hardness of $\SVP$ in the $\ell_1$ norm.)

\subsection{Technical overview}\label{sec:tech-overview}
\subsubsection{Parameterized inapproximability of $\gamma$-$\MDP_q$}
\label{sec:mdp-overview}
Inapproximability results for $\MDP$ and $\SVP$ follow the blueprint originally pioneered by Ajtai~\cite{DBLP:conf/stoc/Ajtai98}, Micciancio~\cite{journals/siamcomp/Micciancio00}, and Khot~\cite{Kho05} for lattices and Dumer, Micciancio, and Sudan~\cite{journals/tit/DumerMS03} for codes. In each case, the idea is to reduce the affine versions of the problems ($\NCP$ and $\CVP$, respectively), for which $\NP$-hardness results were long known, to the linear versions ($\MDP$ and $\SVP$, respectively).

\medskip\noindent\textbf{The DMS reduction from $\NCP$ to $\MDP$.}
We start by illustrating the Dumer-Micciancio-Sudan (DMS) reduction from $\NCP$ to $\MDP$, which is based on analogous reductions of Ajtai~\cite{DBLP:conf/stoc/Ajtai98} and Micciancio~\cite{journals/siamcomp/Micciancio00} from $\CVP$ to $\SVP$. 
An instance of $\NCP$ consists of a linear code $\C = \C(G) \subset \F_q^m$ generated by a matrix $G \in \F_q^{m \times n}$ and a target $\vec{t} \in \F_q^m$, and the goal is to minimize the distance $\dist(\vec{t},\C)$ of $\vec{t}$ to its closest codeword, i.e., the minimum Hamming weight of $G\vec{x} - \vec{t}$ over all $\vec{x} \in \F_q^n$, where the Hamming weight of a vector $\vec{v}$ is $\norm{\vec{v}}_0=|\{i:\vec{v}_i\neq 0\}|$. 
A natural reduction to MDP will produce the instance $\C' = \text{span}(\C,\vec{t})$ generated by $G'  = ( G \mid \vec{t}) \in \F_q^{m \times (n+1)}$. 
If we restrict to codewords $G\vec{x} + \beta\vec{t}$ of $\C'$ that use the target in the combination, i.e., have $\beta \neq 0$, then the minimum distance of such a codeword equals the Hamming distance $\dist(\vec{t},\C)$ of $\vec{t}$ to $\C$. Under this (unreasonable) restriction we have a reduction that preserves the objective value. The obvious trouble though is that $\C$ (and hence $\C'$) might have short codewords of weight much smaller than $\dist(\vec{t},\C)$. In this case, the minimum distance of $\C'$ will equal the distance of $\C$, and have nothing to do with $\vec{t}$. Note, however, that this reduction \emph{does} work if $\lambda(\C) > k$, where $\lambda(\C)=\min_{\vec{c}\in\C\setminus\{\vec{0}\}}\norm{\vec{c}}_0$ is the minimum Hamming weight of $\C$. 
Further, starting from a gap-$\gamma$ version of $\NCP$ asking if $\dist(\vec{t},\C) \le k$ or $\dist(\vec{t},\C) > \gamma k$, we would get hardness of a gap-$\gamma$ version of $\MDP$ if $\lambda(\C) > \gamma k$. 

A natural goal is therefore to increase the distance of $\C$ without increasing the proximity parameter in $\NCP$ by too much. This was achieved in~\cite{journals/tit/DumerMS03} by encoding the message
according to $\C$ as well as a second code $\tilde{\C} \subset \F_q^{m'}$ with generator matrix $\tilde{G} \in \F_q^{m' \times n'}$ with large distance, say $D$.
Further, $\tilde{\C}$ will be a \emph{locally dense code} in the sense that one can find a ``bad list decoding configuration" comprising a center $\vec{s} \in \F_q^{m'}$ that has a large number of codewords of $\tilde{\C}$ within distance $\alpha D$ for some $\alpha < 1$; we call $\alpha$ the \emph{relative radius} of the locally dense code. (One can in fact efficiently construct such locally dense codes with any constant relative radius $\alpha > 1/2$ using randomness~\cite{journals/tit/DumerMS03}.) The number of codewords will be so large that one can sample a linear map $T$ that with high probability projects these codewords \emph{onto} $\F_q^n$. If $\tilde{G}$ is the generator matrix of the locally dense code $\tilde{\C}$, the reduction, which will use randomness to pick both the center $\vec{s}$ and the projection $T$, will produce the instance of $\MDP$ generated by 
\begin{equation}
\label{eq:DMS}
 \begin{pmatrix}
G T \tilde{G}  & \vec{t} \\
\tilde{G} & \vec{s}
\end{pmatrix} .
\end{equation}
The completeness of the reduction follows because for any $\vec{x} \in \F_q^n$ that might satisfy $\norm{G\vec{x}-\vec{t}}_0 \le k$, there will be a codeword $\tilde{G} \vec{y} \in \tilde{\C}$ within distance $\alpha D$ from $\vec{s}$ that projects to $\vec{x}$ under $T$. Thus multiplying the generator matrix in \cref{eq:DMS} by $(\vec{y}^T, -1)^T$
will yield a nonzero codeword of weight at most $k + \alpha D$.  Since the distance of $\tilde{\C}$ is $D$, codewords which don't use the last column of \cref{eq:DMS} will have Hamming weight at least $D$. If $\alpha D + k < D$, which is possible to ensure provided $\alpha < 1$, we get a gap. 

\medskip\noindent\textbf{Challenges in the FPT setting.}
It is reasonable to wonder whether the DMS reduction above works directly in the FPT setting.
However, as already pointed out in~\cite[Section 2.1]{BBEGSLMM21}, one quickly runs into some obstacles.
Indeed, the locally dense codes used in~\cite{journals/tit/DumerMS03} have minimum distance $D$ which depends on the input code dimension, and this is necessary to ensure that we can sample the linear map $T$ with the desired properties.
This is because the existence of $T$ implies that there must be at least $|\F_q^n|=q^n$ codewords in $\tilde{\C}$ of Hamming weight at most $\alpha D$.
Since the distance threshold of the resulting $\MDP$ instance is $k'=\alpha D+k$, it follows that $k'$ depends on the input code dimension $n$, and so the DMS reduction is not FPT.

To overcome these issues,~\cite{BBEGSLMM21} modify both the problem they reduce from as well as the reduction itself.
First, instead of reducing from $\NCP$ to $\MDP$, they reduce from a variant of $\NCP$ they call the \emph{Sparse Nearest Codeword Problem} (SNCP), where the Hamming weight of the coefficient vector realizing the nearest codeword is also taken into account.
More precisely, the objective function $\dist(\vec{t},\C)$ of $\NCP$ is replaced by
\begin{equation*}
    \min_{\vec{x}\in\F_q^n}(\|G\vec{x}-\vec{t}\|_0 +\|\vec{x}\|_0),
\end{equation*}
where $\C=\C(G)$.
It is not hard to reduce $\NCP$ to SNCP in the FPT setting, and this allows~\cite{BBEGSLMM21} to avoid having to sample the linear map $T$.
Second, they replace locally dense codes by another variant which they call \emph{locally suffix dense codes} (LSDCs).
These are codes $\tilde{\C}\subseteq\F_q^{m'}$ with minimum distance $D$ and generator matrix
\begin{equation*}
    \begin{pmatrix}
I_n   & 0 \\
\tilde{G}_1 & \tilde{G}_2
\end{pmatrix} \in\F_q^{m'\times n'}
\end{equation*}
which have the property that, given any prefix $\vec{p}\in\F_q^{n}$, for most ``suffix centers'' $\vec{s}\in\F_q^{m'-n}$ there is a suffix $\vec{u}\in\F_q^{m'-n}$ within Hamming distance $\alpha D$ of $\vec{s}$ such that $(\vec{p},\vec{u})\in\tilde{\C}$.
With the help of these notions,~\cite{BBEGSLMM21} consider the $\MDP$ instance generated by
\begin{equation*}
 \begin{pmatrix}
G   & 0 & \vec{t} \\
I_n & 0 & 0\\
\tilde{G}_1 & \tilde{G}_2 & \vec{s}
\end{pmatrix},
\end{equation*}
where $\vec{s}$ is sampled uniformly at random from $\F_q^{m'-n}$.
The proof that this reduction works is similar to the one sketched above for the DMS reduction.
The main challenge is to efficiently construct LSDCs with appropriate parameters, in particular with minimum distance $D$ \emph{independent} of $m'$, $n'$, and $n$.
Unfortunately, known constructions of locally dense codes do not yield LSDCs with the desired properties.
In the binary setting $q=2$,~\cite{BBEGSLMM21} showed that one can take $\tilde{C}$ to be a binary BCH code~\cite{Hoc59,BC60} with design minimum distance $D$.
This ingenious approach allows them to prove that $\gamma$-$\MDP_2$ is $\W[1]$-hard.

It is instructive to discuss more precisely why the choice of binary BCH codes as LSDCs works, and why it cannot be extended to other finite fields.
Binary BCH codes with minimum distance $D$ have codimension $\approx
    \left\lfloor \frac{D-1}{2}\right\rfloor \log(m'+1)$.
The crucial fact that makes the counting analysis of~\cite{BBEGSLMM21} go through is that $\left\lfloor \frac{D-1}{2}\right\rfloor$ is also the \emph{unique decoding radius} for the binary BCH code, i.e., Hamming balls of this radius centered on BCH codewords are disjoint.
In other words, binary BCH codes almost meet the sphere packing bound.
One would hope that replacing binary BCH codes with $q$-ary BCH codes would suffice to show $\W[1]$-hardness of $\gamma$-$\MDP_q$ more generally.
However, $q$-ary BCH codes with minimum distance $D$ have codimension $\approx \lfloor (D-1)(1-1/q)\rfloor\log_q(m'+1)$ (see \cref{thm:bch}), \emph{while the unique decoding radius remains $\left\lfloor \frac{D-1}{2}\right\rfloor$}.
Put differently, $q$-ary BCH codes for $q>2$ are no longer close to the sphere packing bound, which breaks the analysis from~\cite{BBEGSLMM21}.  In fact,  for $q > 2$, it is not known if there exist $q$-ary codes with
rate vs.\ distance trade-off close to the sphere packing bound. 
Therefore, it seems challenging to make the approach from~\cite{BBEGSLMM21} work as is over $\F_q$, for $q>2$.

\medskip\noindent\textbf{Our approach: Khot for codes.}
We succeed in overcoming the barriers that~\cite{BBEGSLMM21} faced and establish the $\W[1]$-hardness of $\gamma$-$\MDP_q$ for arbitrary finite fields $\F_q$ via a different and arguably simpler (direct) reduction from $\NCP$ to $\MDP$.
Our key insight is to adapt Khot’s reduction~\cite{Kho05} from $\CVP$ to $\SVP$ to the coding-theoretic setting. We are able to meet the requirements of such a reduction with locally dense codes constructed from $q$-ary BCH codes.

This approach is quite natural. In fact, early lecture notes of Khot~\cite{Kho04} showed how to use this strategy to reduce from $\NCP$ to $\MDP$ in the special case of binary codes. Our reduction is more general and requires more careful analysis in that it works with arbitrary locally dense codes with constant relative radius $\alpha < 1$, works over $\F_q$ and not just $\F_2$, and requires a more careful analysis of the the distance bound $k'$ in the output $\MDP$ instance as a function of the distance bound $k$ in the input $\NCP$ instance. However, the core idea is the same.

More precisely, given an instance $(G,\vec{t},k)$ of $\gamma$-$\NCP_q$ with $G\in\F_q^{m\times n}$ and $\vec{t}\in\F_q^m$ and an appropriate locally dense code $(\tilde{G},\vec{s})$ with $\tilde{G}\in\F_q^{m'\times n'}$ and $\vec{s}\in\F_q^{m'}$, we consider the intermediate code $\Cint$ spanned by the generator matrix
\begin{equation*}
    \Gint = \begin{pmatrix}
        G & 0 & -\vec{t}\\
        0 & \tilde{G} & -\vec{s}
    \end{pmatrix}.
\end{equation*}
This is analogous to the intermediate lattice introduced in Khot's reduction~\cite{Kho05} from $\CVP$ to $\SVP$, with the difference being that we replace the $\CVP$ instance by an $\NCP$ instance and the locally dense \emph{lattice} by a locally dense \emph{code}.
Note that it may happen that $\Cint$ contains low weight vectors even when $(G,\vec{t},k)$ is a NO instance of $\gamma$-$\NCP_q$.
This is, however, not a show-stopper, as it in fact suffices to show that there are \emph{many more} low weight vectors in $\Cint$ when $(G,\vec{t},k)$ is a YES instance of $\gamma$-$\NCP_q$ than when $(G,\vec{t},k)$ is a NO instance.
Indeed, if this holds then we can \emph{sparsify} $\Cint$ by intersecting it with an appropriate random code $\Crand$ so that, with high probability, all low weight vectors are eliminated in the NO case, but at least one low weight vector survives in the YES case.
Again, this is analogous to the lattice sparsfication performed in Khot's reduction~\cite{Kho05},
Finally, the $\gamma'$-$\MDP_q$ instance is obtained by computing a generator matrix $\Gfinal$ of $\Cfinal=\Cint\cap\Crand$ and outputting $(\Gfinal, k')$ for some appropriate $k'$.

To guarantee that the reduction is FPT, we need to ensure that $k'\leq f(k)$ for some function $f$.  In fact, in our reduction $k$ only increases by a linear factor, i.e., we get $k' \leq f(k) = O(k)$.
We briefly sketch how to establish the desired properties of $\Cint$ and choose $k'$.
Suppose that $(\tilde{G},\vec{s})$ is a locally dense code with minimum distance $D$ and such that there are at least $N$ vectors $\vec{y}$ satsifying $\|\tilde{G}\vec{y}-\vec{s}\|_0\leq \alpha D$ for some $\alpha\in(1/2,1)$.
If $(G,\vec{t},k)$ is a YES instance of $\gamma$-$\NCP_q$, i.e., there exists $\vec{x}$ such that $\|G\vec{x}-\vec{t}\|_0\leq k$, then multiplying $\Gint$ by $(\vec{x},\vec{y},1)^T$ yields a codeword of weight at most $k'= \alpha D+k$.
As a result, there are at least $N$ vectors in $\Cint$ of weight at most $k'$, which we call \emph{good}.
On the other hand, if $(G,\vec{t},k)$ is a NO instance of $\gamma$-$\NCP_q$ and $D> \gamma k$, then it is not hard to see that every codeword of weight at most $\gamma k$ in $\Cint$ is of the form $\Gint (\vec{x},\vec{0},0)^T$, and so there are at most $q^n$ such \emph{annoying} vectors in $\Cint$, where $n=\dim(\C(G))$. (The ``good'' versus ``annoying'' vectors terminology was also introduced in~\cite{Kho05}.) 

We conclude that the reduction works and is FPT if we are able to construct a $q$-ary locally dense code $(\tilde{G},\vec{s})$ as above under the constraints that (i) $D=g(k)>\gamma k$, (ii) $k'=\alpha D+k\ll \gamma k$, and (iii) $N\gg q^n$ (so that there many more good vectors in YES instances than annoying vectors in NO instances and the sparsification step works).
While the approach of~\cite{BBEGSLMM21} described above required $q$-ary codes of codimension $\approx \frac{D}{2}\log_q(m')$, which are not known to exist for $q>2$, we show that to construct locally dense codes satisfying our constraints it is enough to consider $q$-ary codes with minimum distance $D\approx \gamma k$ and codimension $\approx \beta D\log_q(m')$ for \emph{any} $\beta<1$!
Therefore, we can use $q$-ary BCH codes of length $m'=\poly(m)$ with design minimum distance $D\approx \gamma k$, which have codimension $\approx D(1-1/q)\log_q(m')$ for any prime power $q$.

This approach shows $\W[1]$-hardness of $\gamma'$-$\MDP_q$ for some approximation factor $\gamma'>1$; in fact, we get hardness with $\gamma' \approx 1/\beta \approx 1/\alpha$.
We can then amplify this approximation factor $\gamma'$ in a standard manner via \emph{tensoring} to obtain $\W[1]$-hardness of $\gamma''$-$\MDP_q$ for every $\gamma''\geq 1$.
For more details, see \cref{sec:NCPtoMDP}.

\subsubsection{Parameterized inapproximability of $\gamma$-$\SVP_p$}
\label{sec:svp-overview}

We first define locally dense lattices, which are analogous objects to locally dense codes, and which are important both for understanding the issues with~\cite{BBEGSLMM21} and our ways of handling them. A \emph{locally dense lattice} (with respect to the $\ell_p$ norm) is a lattice $\lat \subset \R^m$ together with a shift $\vec{s} \in \R^m$ such that $\lat - \vec{s}$ contains many vectors of $\ell_p$ norm at most $\alpha \mindistp(\lat)$ for some constant $\alpha = \alpha(p) \in (1/2, 1)$, where $\mindistp(\lat)=\min_{\vec{v}\in\lat\setminus\{\vec{0}\}}\|\vec{v}\|_p$. 
As is the case for locally dense codes, we call $\alpha$ the \emph{relative radius} of the corresponding locally dense lattice.

As in \cref{sec:mdp-overview}, we start by explaining the original approach from~\cite{BBEGSLMM21} towards showing $\W[1]$-hardness of $\gamma$-$\SVP_p$ for $p>1$ and some approximation factor $\gamma>1$, and why it fails to resolve the problems we tackle.
To recall,~\cite{BBEGSLMM21} proved that $\gamma$-$\CVP_p$ is $\W[1]$-hard for every fixed $p,\gamma\geq 1$.
Then, they simply noted that Khot's initial reduction~\cite{Kho05} from $\CVP_p$ to $\SVP_p$ (which is similar to our FPT reduction from $\NCP$ to $\MDP$ discussed in \cref{sec:mdp-overview}) is itself an FPT reduction if the parameters of the locally dense lattice from~\cite{Kho05} (which is based on binary BCH codes) are chosen appropriately.
Combining this observation with the $\W[1]$-hardness of $\gamma$-$\CVP_p$ immediately yields that $\gamma'$-$\SVP_p$ is $\W[1]$-hard for some approximation factor $\gamma'=\gamma'(p)>1$.
However, despite achieving this nice result, the approach of~\cite{BBEGSLMM21} has two significant shortcomings.

\medskip\noindent\textbf{Showing inapproximability of $\SVP$ in all $\ell_p$ norms (including $\ell_1$).}
The first limitation of the approach in~\cite{BBEGSLMM21} is the use of Khot's locally dense lattices, which do not suffice to show either $\NP$- or $\W[1]$-hardness of $\SVP$ in the $\ell_1$ norm.
More specifically, Khot's locally dense lattices have relative radius $\alpha = \alpha(p) > (1/2 + 1/2^p)^{1/p}$, which suffices to show $\W[1]$-hardness of parameterized $\gamma$-$\SVP_p$ (and $\NP$-hardness of non-parameterized $\gamma$-$\SVP_p$) for any
\[
\gamma' = \gamma'(p) < 1/\alpha(p) < (1/2 + 1/2^p)^{-1/p} \ \text{,}
\]
and no better. Plugging $p = 1$ into the right-hand side of this equation shows that Khot's reduction does not yield hardness even for exact $\SVP_1$ (i.e., for $\gamma'$-$\SVP_1$ with $\gamma' = 1$). Indeed, this issue is what kept Khot's reduction from showing $\NP$-hardness of $\SVP_1$ in~\cite{Kho05} and what kept~\cite{BBEGSLMM21} from showing $\W[1]$-hardness of $\SVP_1$.

Despite Khot's reduction not working, other reductions nevertheless showed $\NP$-hardness of (even approximating) $\SVP_1$. 
Unfortunately, as~\cite{BBEGSLMM21} notes, these reductions fail both because of their use of non-integral lattices and the fact that rounding real-valued lattice bases to integral ones in a black-box way amounts to a non-FPT reduction, since the minimum distance of the resulting lattices will depend on their dimension. (Multiplying rational lattice bases by the least common multiple of their entries' denominators causes a similar problem.)
First, Micciancio~\cite{journals/siamcomp/Micciancio00} showed hardness of $\gamma$-$\SVP_1$ for any $\gamma < 2$ using locally dense lattices constructed from prime number lattices. However, these locally dense lattices are non-integral and even non-rational.\footnote{We note in passing that Micciancio did in fact carefully analyze rounding these locally dense lattices to get integral ones, but emphasize again that this rounding causes the minimum distance of the resulting lattices to depend on their dimension.}
Second, Regev and Rosen~\cite{conf/stoc/RegevR06} showed how to use efficiently computable linear norm embeddings to reduce $\gamma$-$\SVP_2$ to $\gamma'$-$\SVP_p$ for any $p \geq 1$ and any constant $\gamma' < \gamma$. 
Combined with Khot's work~\cite{Kho05}, which showed $\NP$-hardness of approximating $\SVP_2$ to within any constant factor $\gamma$,~\cite{conf/stoc/RegevR06} implies that $\SVP_p$ for any $p$ (and in particular, $\SVP_1$) is $\NP$-hard to approximate within any constant factor as well.
However, the norm embeddings given in~\cite{conf/stoc/RegevR06} use random Gaussian projection matrices, and therefore output non-integral lattices. Moreover, using different, integral distributions for the projection matrices also does not obviously work.

We overcome this first issue of Khot's locally dense lattices not working in the $\ell_1$ norm by instantiating Khot's reduction with different locally dense lattices.
Specifically, we instantiate Khot's reduction with the locally dense lattices constructed in recent work of Bennett and Peikert~\cite{BP22}, which are built from Reed-Solomon codes. 
These locally dense lattices meet all of the requirements necessary for the proof of \cref{thm:hardness-l1}. Namely, they are efficiently constructible; their base lattices $\lat$ are integral; they can be constructed so that $\mindistp(\lat)$ does not depend on the dimension of the input lattice $\lat$; and for $p \in [1, \infty)$ they have $\ell_p$ relative radius $\alpha(p) \approx 1/2^p < 1$. In particular, they have $\ell_1$ relative radius $\alpha(1) \approx 1/2$ (which is essentially optimal by the triangle inequality), and so Khot's reduction shows hardness of $\gamma$-$\SVP_{1}$ for any constant $\gamma < 2$. 
(We again note that the largest approximation factor $\gamma = \gamma(p)$ for which Khot's reduction shows parameterized hardness of $\gamma$-$\SVP_p$ is $\gamma \approx 1/\alpha$, where $\alpha = \alpha(p)$ is the relative radius of the locally dense lattice used, and this is where the bound on the approximation factor $\gamma = \gamma(p)$ in \cref{thm:hardness-l1} comes from.)

\medskip\noindent\textbf{Showing inapproximability of $\gamma$-$\SVP_p$ for all $p >1$ and all $\gamma$.}
The second main shortcoming of the approach in~\cite{BBEGSLMM21} is that it is not clear how to amplify the approximation factor $\gamma > 1$ for which they get $\W[1]$-hardness of $\gamma$-$\SVP_p$ (for $p > 1$), to an arbitrary constant. As in the case of codes, the natural thing to try for amplifying hardness is to take the \emph{tensor product} of the input $\SVP$ instance with itself. The idea of tensoring is, given an instance $(B, k)$ of $\SVP$ as input, to output the $\SVP$ instance $(B \otimes B, k^2)$, where $B \otimes B$ is the Kronecker product of the input basis matrix $B$ with itself. Unfortunately, unlike for codes, tensoring does not work in general for lattices. Indeed, although it always holds that $\lambda_1(\lat(B) \otimes \lat(B)) \leq \lambda_1(\lat(B))^2$, the converse is not always true or even ``close to true''; see, e.g.,~\cite[Lemma 2.3]{HR12}.

Although Haviv and Regev~\cite{HR12} showed that Khot's original $\SVP_2$ instances have properties that \emph{do} in fact allow them to tensor nicely,
this is \emph{not the case} for the $\SVP$ instances obtained in~\cite{BBEGSLMM21}.
Indeed, the (crucial!) subtlety is that the standard $\NP$-hardness proof for approximate $\CVP_p$ proceeds via a reduction from approximate \emph{Exact Set Cover}~\cite{journals/jcss/AroraBSS97,Kho05,HR12}, and the resulting $\CVP_p$ instances enjoy important additional properties that are then inherited by the $\SVP_p$ instances in~\cite{Kho05}.
Parameterized inapproximability of Exact Set Cover is known under the (randomized) Gap-ETH and PIH assumptions, and this is what allowed~\cite{BGKM18} to show parameterized hardness of $\gamma$-$\SVP_2$ for any constant $\gamma \geq 1$; see also~\cref{remark:svp2-tensoring}.
However, it is not currently known whether approximate Exact Set Cover is $\W[1]$-hard, and so~\cite{BBEGSLMM21} generate their $\CVP_p$ instances via a different reduction from (the dual version of) $\NCP_q$ with a suitably large prime $q$ instead.\footnote{The exact version of this problem is known to be $\W[1]$-hard, see~\cite[Section 13.6.3]{CyganBook15}.}
As a result, important properties no longer hold when~\cite{BBEGSLMM21} use these alternative $\CVP_p$ instances to create $\SVP_p$ instances via Khot's reduction.
Namely, it is no longer true that every lattice vector with at least one odd coordinate has large Hamming weight, a property that is needed to ensure that the $\SVP_p$ instance tensors nicely in~\cite{HR12}.

It is also sensible to wonder whether Khot's \emph{augmented tensor product}~\cite{Kho05}, which he introduced in his original work to overcome issues with tensoring, can nevertheless be used to boost the approximation factor of the $\SVP_p$ instances generated in~\cite{BBEGSLMM21}.
However, the augmented tensor product cannot be applied in the FPT setting unless the short lattice vectors in the base $\SVP$ instances also have have short coefficient vectors (i.e., coefficient vectors whose $\ell_p$ norm is independent of lattice dimension). The $\SVP$ instances in~\cite{BBEGSLMM21} do not
seem to have this property.

\medskip\noindent\textbf{Our solution.}
In order to construct $\W[1]$-hard $\SVP$ instances that tensor nicely and thereby prove \cref{thm:hardness-lp-tensor}, we give a reduction directly from approximate $\NCP_2$ to approximate $\SVP_p$ for any $p >1$.\footnote{We could also use $\CVP$ as an intermediate problem in the reduction as is done in~\cite{BBEGSLMM21}, but that does not obviously make the reduction simpler or more modular.}
Our reduction is a variant of the reductions in Khot~\cite{Kho05} and Haviv and Regev~\cite{HR12}, and again we instantiate the reduction with locally dense lattices constructed from binary BCH codes similar to those used by Khot~\cite{Kho05,BBEGSLMM21}.
We emphasize that although the proofs of~\cref{thm:hardness-l1,thm:hardness-lp-tensor} both use variants of Khot's reduction, the key to proving~\cref{thm:hardness-l1} was to instantiate Khot's reduction with different locally dense lattices and the key to~\cref{thm:hardness-lp-tensor} was to reduce from a different $\W[1]$-hard problem.
Moreover, ensuring that the characteristic of the underlying codes in the $\NCP$ instances that we reduce from matches that of the underlying BCH codes in the locally dense lattices that we use seems essential for our analysis. Indeed, our $\NCP$ instances and locally dense lattices both use codes over $\F_2$, whereas~\cite{BBEGSLMM21} reduced from $\NCP$ instances over $\F_q$ for larger prime $q$.

Our reduction allows us to construct $\gamma$-$\SVP_2$ instances with some constant $\gamma > 1$ that meet the sufficient conditions given in~\cite{HR12} to be amplified to $\gamma'$-$\SVP_2$ instances for arbitrarily large constant $\gamma'$. These conditions roughly say that the base lattices $\lat$ in the $\SVP$ instance must be such that all vectors $\vec{v} \in \lat \subseteq \Z^n$ satisfy at least one of the following: (1) $\vec{v}$ has Hamming weight at least $d$ for some distance bound $d$, (2) $\vec{v} \in 2\Z^n$ and $\vec{v}$ has Hamming weight at least $d/4$, or (3) $\vec{v} \in 2\Z^n$ and $\vec{v}$ has very high $\ell_2$ norm. The minimum distances of lattices in which all vectors satisfy either conditions (1) or (2) behave nicely under tensoring, but condition (3) makes the analysis subtle. However, we are essentially able to rely on the analysis in~\cite{HR12}. Moreover, modifying the analysis in~\cite{HR12} a bit additionally allows us to extend our result to $\ell_p$ norms for $p >1$.
(The omission of $p = 1$ is yet again for the same reason as in~\cite{Kho05,BBEGSLMM21}; it is because of the binary BCH-code-based locally dense lattices that we use.)

\subsection{Additional related work}
\label{sec:related-work}

Interest in the complexity of computational problems on codes and lattices more broadly goes back several decades.
We survey the most closely related work here.

\medskip\noindent\textbf{Complexity of $\NCP$ and $\MDP$.}
Berlekamp, McEliece, and van Tilborg~\cite{BMvT78} showed that certain problems related to linear codes, such as the exact version of $\NCP$, are $\NP$-hard.
They also conjectured that the exact version of $\MDP$ is $\NP$-hard.
This conjecture remained open until groundbreaking work of Vardy~\cite{Var97}, who showed that exact $\MDP$ is indeed $\NP$-hard.
Not long after, Dumer, Micciancio, and Sudan~\cite{journals/tit/DumerMS03} showed that approximate $\MDP$ is $\NP$-hard under randomized reductions.
Follow-up work by Cheng and Wan~\cite{CW12}, Austrin and Khot~\cite{AK14}, and Micciancio~\cite{Mic14} showed that approximate $\MDP$ is $\NP$-hard under deterministic reductions.
The unparameterized fine-grained hardness of $\NCP$ and $\MDP$ was recently studied by Stephens-Davidowitz and Vaikuntanathan~\cite{SV19}. 

On the parameterized front, Downey, Fellows, Vardy, and Whittle~\cite{DFVW99} showed, among other things, that the exact version of $\NCP$ is $\W[1]$-hard, and infamously conjectured that $\MDP$ is $\W[1]$-hard.
As discussed above, the status of this conjecture did not budge until the seminal work~\cite{BBEGSLMM21}, where it was shown that $\gamma$-$\NCP_q$ is $\W[1]$-hard for every $\gamma\geq 1$ and prime power $q$, and that $\gamma$-$\MDP_2$ is $\W[1]$-hard for every $\gamma\geq 1$.
Finally, by establishing the parameterized fine-grained hardness of Exact Set Cover and invoking results from~\cite{BGKM18,BBEGSLMM21}, Manurangsi~\cite{Man20} showed that there are no algorithms running in time $n^{o(k)}$ for deciding $\gamma$-$\NCP_q$ (for all constant $\gamma\geq 1$) and $\gamma$-$\MDP_2$ (for some constant $\gamma>1$) assuming Gap-ETH.

\medskip\noindent\textbf{Complexity of $\CVP$ and $\SVP$.}
The study of the complexity of lattice problems was initiated by van Emde Boas~\cite{vanEmdeBoas81}, who showed that $\CVP_2$ was $\NP$-hard. He also showed that $\SVP_{\infty}$ is $\NP$-hard and conjectured that $\SVP_2$ was $\NP$-hard.
This result remained the state-of-the-art until Ajtai~\cite{DBLP:conf/stoc/Ajtai98} extended it to the $\ell_2$ norm under randomized reductions, and a deep line of work soon followed showing progressively stronger hardness of approximation results for $\SVP_p$ in different $\ell_p$ norms~\cite{conf/coco/CaiN98,journals/siamcomp/Micciancio00,journals/tcs/Dinur02,Kho05,HR12,Mic12}.
A recent line of work has also focused on the (unparameterized) fine-grained hardness of approximate $\CVP$ and $\SVP$~\cite{conf/focs/BennettGS17,conf/stoc/AggarwalS18,conf/soda/AggarwalBGS21,conf/innovations/BennettPT22}.

In terms of parameterized hardness, Downey, Fellows, Vardy, and Whittle~\cite{DFVW99} showed that exact $\CVP$ is $\W[1]$-hard, and asked whether $\SVP$ is $\W[1]$-hard.
As was the case for $\MDP$, this question was only settled in~\cite{BBEGSLMM21}, where it was shown that $\gamma$-$\CVP_p$ is $\W[1]$-hard for all $p\geq 1$ and $\gamma\geq 1$, and that $\gamma$-$\SVP_p$ is $\W[1]$-hard for $p>1$ with some $\gamma = \gamma(p) >1$.
From a fine-grained perspective, it was shown by Manurangsi~\cite{Man20} that, assuming Gap-ETH, there are no algorithms running in time $n^{o(k^p)}$, where $n$ is the rank of the input lattice, for deciding $\gamma$-$\CVP_p$ with any $\gamma\geq 1$ for $p\geq 1$ and deciding $\gamma$-$\SVP_p$ with some $\gamma>1$ for all $p>1$.

\subsection{Open problems}
\label{sec:open-problems}

We highlight two interesting directions for future research:
\begin{itemize}
    \item The reductions that we use to prove all of our main theorems are randomized and have two-sided error due to our randomized constructions of locally dense codes and lattices and due to our use of sparsification. It would be a groundbreaking contribution to find ways to derandomize these reductions and obtain deterministic parameterized hardness results for $\MDP$ and $\SVP$.
    We note that when it comes to showing $\NP$-hardness (instead of $\W[1]$-hardness), we know deterministic reductions from $\NCP$ to $\MDP$~\cite{CW12,AK14,Mic14} and randomized reductions with one-sided error from $\CVP$ to $\SVP$~\cite{Mic12}. Additionally, we note that showing deterministic ($\NP$-)hardness of $\SVP$ in the non-parameterized setting is a major open question.
    
    \item We have shown that $\gamma$-$\SVP_p$ is $\W[1]$-hard for any fixed $p>1$ and $\gamma\geq 1$.
    When $p=1$, we showed that $\gamma$-$\SVP_p$ is $\W[1]$-hard when $\gamma\in[1,2)$.
    We leave it as a fascinating open problem to extend our $\W[1]$-hardness result for all $\gamma\geq 1$ to $p=1$ as well.
    This is an important missing piece of our understanding of the parameterized hardness of approximate $\SVP$ in $\ell_p$ norms.
\end{itemize}

\fullornot{
\subsection{Acknowledgements}
We thank Arnab Bhattacharyya, Ishay Haviv, Bingkai Lin, Pasin Manurangsi, Xuandi Ren, and Noah Stephens-Davidowitz for helpful comments and answers to our questions. In particular, we would like to thank Ishay Haviv for sketching how to extend the tensoring-based hardness amplification for $\SVP$ in~\cite{HR12} to general $\ell_p$ norms, and would like to thank Pasin Manurangsi for clarifying the status of $\W[1]$-hardness of approximation results for $\SVP$ in the $\ell_2$ norm.

H.\ Bennett's research was partially supported by NSF Award No.\ CCF-2312297.
M.\ Cheraghchi's research was partially supported by the National Science Foundation under Grants No.\ CCF-2006455, CCF-2107345,
and a CAREER Award CCF-2236931.
V.\ Guruswami's research was supported in part by NSF grants CCF-2228287 and CCF-2210823, a Simons Investigator award, and a UC Noyce Initiative award.
J.\ Ribeiro's research was supported by NOVA LINCS (UIDB/04516/2020) with the financial support of FCT - Fundação para a Ciência e a Tecnologia
and by the NSF grants CCF-1814603 and CCF-2107347 and the following grants of Vipul Goyal: the NSF award 1916939, DARPA SIEVE program, a gift from Ripple, a DoE
NETL award, a JP Morgan Faculty Fellowship, a PNC center for financial services innovation award, and a
Cylab seed funding award.
}{}
\section{Preliminaries}
\label{sec:prelims}

Throughout we use boldface, lower-case letters like $\vec{v}, \vec{x}, \vec{s}, \vec{t}$ to denote column vectors.

\subsection{Probability theory}

We denote random variables by uppercase letters such as $X$, $Y$, and $Z$.
Throughout this work we consider only discrete random variables supported on finite sets.
Given a random variable $X$, we denote its expected value by $\E[X]$ and its variance by $\Var[X]$.
We write the indicator random variable for an event $E$ as $\mathbf{1}_{\{E\}}$.

We will make use of the following standard corollary of Chebyshev's inequality.
For completeness, we provide a short proof.
\begin{lemma}\label{lem:chebyshev}
Let $X_1,\dots,X_N$ be pairwise independent random variables over $\{0,1\}$ such that $\Pr[X_i=1]=p>0$ for $i=1,\dots,N$. 
Then, it holds that
\begin{equation*}
    \Pr\left[\forall i\in[N], X_i=0\right]\leq \frac{1}{pN}.
\end{equation*}
\end{lemma}
\begin{proof}
    Let $X=\sum_{i=1}^N X_i$ and note that $\E[X]=pN$.
    We have
    \begin{equation*}
        \Pr\left[X=0\right] \leq \Pr\left[\left|X-\E[X]\right|\geq pN\right]
        \leq \frac{\Var[X]}{(pN)^2}
        =\frac{(1-p)pN}{(pN)^2}
        \leq \frac{1}{pN},
    \end{equation*}
    where the second inequality follows from Chebyshev's inequality and the equality holds due to the pairwise independence of the $X_i$'s.
\end{proof}

\subsection{Parameterized promise problems and FPT reductions}

We recall basic definitions  related to parameterized promise (decision) problems and Fixed-Parameter Tractable (FPT) reductions between such problems.
We refer the reader to~\cite{DF99} for an excellent discussion on parameterized algorithms and reductions.

\begin{definition}[Parameterized language]

A set $\cS\subseteq\Sigma^*\times \N$ is said to be a \emph{parameterized language} (with respect to the rightmost coordinate).
    
\end{definition}

\begin{definition}[Parameterized promise problem]
    The tuple of parameterized languages $\Pi=(\Pi_{\textnormal{YES}},\Pi_{\textnormal{NO}})$ is said to be a \emph{parameterized promise problem} if $\{x:(x,k)\in\Pi_{\textnormal{YES}}\}\cap \{x:(x,k)\in\Pi_{\textnormal{NO}}\}=\emptyset$ for every parameter choice $k\in\N$.
\end{definition}

\begin{definition}[Randomized FPT reductions with two-sided error]
    We say that a randomized algorithm is a \emph{randomized FPT reduction with two-sided error} from the parameterized promise problem $\Pi$ to the parameterized promise problem $\Pi'$ if the following properties hold:
    \begin{itemize}
        \item On input $(x,k)$, the algorithm runs in time at most $T(k)\cdot |x|^c$ for some computable function $T(\cdot)$ and some absolute constant $c>0$ and outputs a tuple $(x',k')$;
        
        \item It holds that $k'\leq g(k)$ for some computable function $g(\cdot)$;
        
        \item If $(x,k)\in\Pi_{\textnormal{YES}}$, it holds that $\Pr\left[(x',k')\in \Pi'_{\textnormal{YES}}\right]\geq 2/3$, where the probability is taken over the randomness of the algorithm;
        
        \item If $(x,k)\in\Pi_{\textnormal{NO}}$, it holds that $\Pr\left[(x',k')\in \Pi'_{\textnormal{NO}}\right]\geq 2/3$, where the probability is taken over the randomness of the algorithm.
    \end{itemize}
\end{definition}

Note that if there is a randomized FPT reduction with two-sided error from $\Pi$ to $\Pi'$, it follows that there is a randomized FPT algorithm (i.e., an algorithm running in time $T(k)\cdot|x|^c$ for some computable function $T(\cdot)$ on input an instance $(x,k)$) for deciding $\Pi$ with two-sided error whenever there is such an algorithm for deciding $\Pi'$.
The success probability of any such algorithm can be amplified in a standard manner.

In this work we focus on the parameterized complexity class $\W[1]$.
It is well-known that the parameterized Clique problem, in which we are given as input a graph $G$ and a positive integer $k$ (with $k$ being the parameter of interest) and must decide whether $G$ contains a clique of size $k$, is $\W[1]$-complete. 
That is, parameterized Clique is in $\W[1]$ and it is $\W[1]$-hard, i.e., there is an FPT reduction from every problem in $\W[1]$ to it (see, e.g.,~\cite[Theorem 13.18]{CyganBook15}).
Therefore, one may \emph{define} $\W[1]$ to be the class of all parameterized problems with FPT reductions to Clique.
We refrain from discussing $\W[1]$ in more detail; for an extensive discussion, see~\cite[Chapters 9 to 11]{DF99}.

It is widely believed that $\W[1]$ problems cannot be decided by FPT algorithms, even if randomness with two-sided error is allowed.
We say that a parameterized promise problem $\Pi'$ is \emph{$\W[1]$-hard under randomized reductions} if there is a randomized FPT reduction with two-sided error from a $\W[1]$-hard problem $\Pi$ to $\Pi'$.
The existence of such a reduction shows that $\Pi'$ is likely intractable from a parameterized perspective.

\subsection{Coding problems}\label{sec:coding-problems}
Let $\C(G) := \set{G\vec{x} : \vec{x} \in \F_q^n}$ denote the code generated by the generator matrix $G \in \F_q^{m \times n}$ (note that here $\C(G)$ is the $\F_q$-span of the \emph{columns} of $G$).
Alternatively, we may see $\C(G)$ as the kernel of the parity-check matrix $H\in\F_q^{(m-n)\times m}$ which spans the dual subspace of $\C(G)$, i.e., $\C(G)=\{\vec{y}\in\F_q^m:H\vec{y}=\vec{0}\}$.
We call sets of the form $\vec{u}+\C=\{\vec{u}+\vec{c}:\vec{c}\in\C\}$ with $\vec{u}\in\F_q^m$ the cosets of $\C$. 
We write $\|\vec{x}\|_0=|\{i\in[m]:\vec{x}_i\neq 0\}$ for the Hamming weight of a vector $\vec{x}\in\F_q^m$ and call $\|\vec{x}-\vec{y}\|_0$ the Hamming distance between $\vec{x}$ and $\vec{y}$.
For a code $\C \subseteq \F_q^m$, let $\lambda(\C)=\min_{\vec{c}\in\C\setminus\{\vec{0}\}}\|\vec{c}\|_0$ be the Hamming minimum distance of $\C$, and let $\dist(\vec{y}, \C) := \min_{\vec{c} \in \C} \norm{\vec{y} - \vec{c}}_0$ denote the Hamming distance between a vector $\vec{y} \in \F_q^m$ and $\C$.  Let $\B_{q, m}(r) =\{\vec{x}\in\F_q^m: \|\vec{x}\|_0\leq r\}$ denote the Hamming ball of radius $r$ in $\F_q^m$.

We define two fundamental promise problems from coding theory.

\begin{definition}[Nearest Codeword Problem]\label{def:ncp}
The \emph{$\gamma$-approximate Nearest Codeword Problem over $\F_q$ ($\gamma$-$\NCP_q$)} is the decisional promise problem defined as follows.
On input a generator matrix $G \in \F_q^{m \times n}$, target $\vec{t} \in \F_q^m$, and distance parameter $k \in \Z^+$, the goal is to decide between the following two cases when one is guaranteed to hold:
\begin{itemize}
    \item (YES) $\dist(\C(G), \vec{t}) \leq k$,
    \item (NO) $\dist(\C(G), \vec{t}) > \gamma k$.
\end{itemize}
The parameter of interest is $k$.
\end{definition}

\begin{remark}
A scaling argument shows that the NO case in \cref{def:ncp} is equivalent to
\begin{itemize}
    \item (NO) $\dist(\C(G), \alpha\vec{t}) > \gamma k$ for any $\alpha\in\F_q\setminus\{0\}$.
\end{itemize}
\end{remark}

The following results establish the $\W[1]$-hardness and parameterized fine-grained hardness of NCP.
\begin{theorem}[\protect{\cite[Theorem 5.1, adapted]{BBEGSLMM21}}]\label{thm:hardNCP}
For any prime power $q\geq 2$ and real number $\gamma\geq 1$ it holds that $\gamma$-$\NCP_q$ is $\W[1]$-hard.
\end{theorem}

\begin{theorem}[\protect{\cite[Corollary 5, adapted]{Man20}}]\label{thm:FGhardnessNCP}
For any fixed prime power $q$ and $\gamma\geq 1$ and any function $T$, assuming randomized Gap-ETH, there is no randomized algorithm running in time $T(k) n^{o(k)}$ which decides $\gamma$-$\NCP_q$ with probability at least $2/3$, where $n$ is the dimension of the input code.
\end{theorem}

We remark that~\cite{BBEGSLMM21} states \cref{thm:hardNCP} as the $\W[1]$-hardness of the ``$\gamma$-MLD$_q$'' problem (where ``MLD'' stands for ``Maximum Likelihood Decoding'' and the parameter of interest is again the input distance $k$), which is equivalent to the $\gamma$-$\NCP_q$ problem.
More precisely, the input to the $\gamma$-MLD$_q$ problem consists of a parity-check matrix $H\in\F_q^{h\times m}$, a target $\vec{t}\in\F_q^m$, and a distance bound $k$, and we must decide whether there exists a vector $\vec{e}$ with $\|\vec{e}\|_0\leq k$ such that $H\vec{e} = H\vec{t}$ (i.e., $\vec{t}$ and $\vec{e}$ have the same syndrome), or whether all such vectors $\vec{e}$ have Hamming weight larger than $\gamma k$.
This is equivalent to $\gamma$-$\NCP_q$ because we can efficiently compute the generator matrix $G$ of the code with parity-check matrix $H$ and vice-versa, and because $H\vec{e} = H\vec{t}$ if and only if $\vec{t} = \vec{c} + \vec{e}$ for some codeword $\vec{c}\in\C(G)$.
Moreover,~\cite{BBEGSLMM21} only stated the result for prime $q$.
However, direct inspection of~\cite[Section 5.2]{BBEGSLMM21} shows that their proof also yields the more general version stated in \cref{thm:hardNCP}.
In particular,~\cite[Definition 5.3]{BBEGSLMM21}, including the two observations there, generalizes to arbitrary finite fields.

\begin{definition}[Minimum Distance Problem]
The \emph{$\gamma$-approximate Minimum Distance Problem over $\F_q$ ($\gamma$-$\MDP_q$)} is the decisional promise problem defined as follows. On input a generator matrix $G \in \F_q^{m \times n}$ and distance parameter $k \in \Z^+$, the goal is to decide between the following two cases when one is guaranteed to hold:
\begin{itemize}
    \item (YES) $\lambda(\C(G)) \leq k$,
    \item (NO) $\lambda(\C(G)) > \gamma k$.
\end{itemize}
The parameter of interest is $k$.
\end{definition}

\subsubsection{Tensoring codes}
\label{sec:tensoring-codes}

The tensor product of linear codes is an important operation for building new codes with interesting properties by combining two linear codes.
In particular, tensoring can be used to boost the approximation factor in $\W[1]$-hardness results for NCP and MDP from some constant $\gamma>1$ to an arbitrary constant.

Given two linear codes $\C(G_1)$ and $\C(G_2)$ with $G_i\in\F_q^{m_i\times n_i}$ and minimum distance $d_i$ for $i=1,2$, we define the associated tensor product code as
\[
\C(G_1) \otimes \C(G_2) := \C(G_1\otimes G_2) \ \text{,}
\]
where  $G_1\otimes G_2\in\F_q^{m_1 m_2\times n_1 n_2}$ is the Kronecker product of $G_1$ and $G_2$.
Furthermore, we have
\begin{equation} \label{eq:tensor-code-min-dist}
    \lambda(\C(G_1) \otimes \C(G_2)) = d_1\cdot d_2 \ \text{.}
\end{equation}
See, e.g.,~\cite[Section V.B]{journals/tit/DumerMS03} for a proof.

Suppose that we know that $\gamma$-$\MDP_q$ is $\W[1]$-hard (under randomized reductions) for \emph{some} $\gamma>1$.
Then, using \cref{eq:tensor-code-min-dist}, we can immediately conclude that for any integer $c\geq 1$, $\gamma^c$-$\MDP_q$ is $\W[1]$-hard (under randomized reductions) by considering the tensored $\MDP$ instances $(B^{\otimes c},k^c)$, 
where $B^{\otimes c}$ denotes the $c$-fold Kronecker product of $B$ with itself.
In particular, constructing tensored $\MDP$ instances in this way gives an FPT self-reduction from $\gamma$-$\MDP_q$ to $\gamma^c$-$\MDP_q$.

\subsection{Lattice problems}\label{sec:lattice-problems}

Let $\lat(B)=\{B\vec{x}:\vec{x}\in\Z^n\}$ denote the lattice generated by the matrix $B\in\R^{m\times n}$.
We call $n$ the rank of $\lat(B)$ and write $\det(\lat(B)) = \sqrt{\det(B^T B)}$ for the determinant of $\lat(B)$, where $B^T$ denotes the transpose of $B$.
For $p \in [1, \infty)$, we write $\|\vec{x}\|_p=\left(\sum_{i=1}^m \abs{x_i}^p\right)^{1/p}$ for the $\ell_p$ norm of a vector $x\in\R^m$.
We use $\lambda_1^{(p)}(\lat)$ to denote the $\ell_p$ norm of a shortest nonzero vector in $\lat$ and set $\dist_p(\lat,\vec{t}):=\min_{\vec{v}\in\lat}\|\vec{v}-\vec{t}\|_p$.
We write $\B^{(p)}_m(r)=\{\vec{x}\in\R^m:\|\vec{x}\|_p\leq r\}$ for the closed, centered $\ell_p$ ball of radius $r$ in $\R^m$.

We define two fundamental promise problems related to lattices.

\begin{definition}[Closest Vector Problem]
The \emph{$\gamma$-approximate Closest Vector Problem with respect to the $p$-norm ($\gamma$-$\CVP_p$)} is the decisional promise problem defined as follows.
On input a generator matrix $B\in\Z^{m\times n}$, a target $\vec{t}\in\Z^m$, and a distance parameter $k\in\Z^+$, the goal is to decide between the following two cases when one is guaranteed to hold:
\begin{itemize}
    \item (YES) $\dist_p(\lat(B),\vec{t})\leq k$,
    
    \item (NO) $\dist_p(\lat(B),\alpha\vec{t})>\gamma k$ for any $\alpha\in\Z\setminus\{0\}$.
\end{itemize}
The parameter of interest is $k$.
\end{definition}

The definition above is a slight (but widely used) variant of the original Closest Vector Problem, since in the NO case we also require all multiples $\alpha\vec{t}$ with $\alpha\in\Z\setminus\{0\}$ to be far from $\lat(B)$.
The following results about the $\W[1]$-hardness and parameterized fine-grained hardness of this variant of CVP are known to hold.
\begin{theorem}[\protect{\cite[Theorem 7.2]{BBEGSLMM21}}]\label{thm:W1hardnessCVP}
For any real numbers $\gamma,p\geq 1$ it holds that $\gamma$-$\CVP_p$ is $\W[1]$-hard.
\end{theorem}

\begin{theorem}[\protect{\cite[Corollary 6, adapted]{Man20}}]\label{thm:FGhardnessCVP}
For any fixed $p,\gamma\geq 1$ and any function $T$, assuming randomized Gap-ETH, there is no randomized algorithm running in time $T(k) n^{o(k^p)}$ which decides $\gamma$-$\CVP_p$ with probability at least $2/3$, where $n$ is the rank of the input lattice.
\end{theorem}

\begin{definition}[Shortest Vector Problem]
The \emph{$\gamma$-approximate Shortest Vector Problem with respect to the $\ell_p$-norm ($\gamma$-SVP$_p$)} is the decisional promise problem defined as follows.
On input a generator matrix $B\in\Z^{m\times n}$ and a distance parameter $k\in\Z^+$, the goal is to decide between the following two cases when one is guaranteed to hold:
\begin{itemize}
    \item (YES) $\lambda_1^{(p)}(\lat(B))\leq k$,
    
    \item (NO) $\lambda_1^{(p)}(\lat(B))>\gamma k$.
\end{itemize}
The parameter of interest is $k$.
\end{definition}

\subsubsection{Tensoring lattices}
\label{sec:tensoring-lattices}
Analogously to the coding setting, we can also consider the tensor product of lattices.
Given two lattices $\lat(B_1)$ and $\lat(B_2)$ with basis matrices $B_1\in\Z^{m_1\times n_1}$ and $B_2\in\Z^{m_2\times n_2}$, we define the associated tensor product lattice as
\begin{equation*}
    \lat(B_1)\otimes \lat(B_2) := \lat(B_1\otimes B_2) \ \text{,}
\end{equation*}
where  $B_1\otimes B_2\in\Z^{m_1 m_2\times n_1 n_2}$ is the Kronecker product of $B_1$ and $B_2$.
The tensor product lattice is independent of the bases we choose for the two underlying lattices.

Unlike for codes, it is not true that repeated tensoring of lattices allows us to generically boost the approximation factor in known hardness results for $\CVP$ and $\SVP$.
Indeed, while it always holds that
\begin{equation*}
    \lambda_1^{(p)}(\lat(B_1)\otimes\lat(B_2))\leq \lambda_1^{(p)}(\lat(B_1))\cdot \lambda_1^{(p)}(\lat(B_2)) \ \text{,}
\end{equation*}
it may happen that the left-hand side of this inequality is significantly smaller than the right-hand side.
For an example, see~\cite[Lemma 2.3]{HR12}.
Therefore, additional effort is required to prove special structural properties of our $\CVP$ and $\SVP$ instances to ensure that tensoring them allows us to boost the approximation factor in our hardness results.

\subsection{Locally dense codes and lattices}

Our randomized FPT reductions from NCP to MDP and from CVP to SVP use families of \emph{locally dense} codes and lattices with appropriate parameters.
Precise definitions of such objects follow below.

\begin{definition}[Locally dense code]
Fix a real number $\alpha\in(0,1)$, positive integers $d, N, m, n$, and a prime power $q$.
A \emph{$(q,\alpha,d, N, m, n)$-locally dense code} is specified by a generator matrix $A \in \F_q^{m \times n}$ and a target vector $\vec{s} \in \F_q^m$ with the following properties:
\begin{itemize}
    \item $\lambda(\C(A)) > d$.
    \item $\card{(\C(A) - \vec{s}) \cap \B_{q, m}(\alpha d)} \geq N$.
\end{itemize}
\end{definition}
That is, the code $\C(A)$ has block length $m$, dimension $n$, (design) minimum distance $d$, is over $\F_q$, and a ``bad list decoding configuration" with at least $N$ codewords within Hamming distance $\alpha d<d$ of $\vec{s}$.

\begin{definition}[Locally dense lattice]
    Fix real numbers $\alpha\in(0,1)$ and $p\geq 1$ and positive integers $d, N, m, n$.
A \emph{$(p,\alpha,d, N, m, n)$-locally dense lattice} is specified by a basis $A \in \Z^{m \times n}$ and a target vector $\vec{s} \in \Z^m$ with the following properties:
\begin{itemize}
    \item $\lambda^{(p)}_1(\lat(A)) > d$.
    \item $\card{(\lat(A) - \vec{s}) \cap \B^{(p)}_{m}(\alpha d)} \geq N$.
\end{itemize}
\end{definition}

\section{The FPT $\NCP_q$ to $\MDP_q$ reduction}\label{sec:NCPtoMDP}

We next describe and analyze a randomized FPT reduction from approximate $\NCP_q$ to approximate $\MDP_q$ which works over any finite field. 
Our reduction is obtained by adapting Khot's reduction~\cite{Kho05,BBEGSLMM21} from approximate CVP to approximate SVP to the coding setting and combining it with locally dense codes constructed with the help of BCH codes over general finite fields.
Combined with \cref{thm:hardNCP}, our reduction yields \cref{thm:hardness-mdp}, which we restate here.

\hardnessmdp*

\subsection{A reduction with advice}\label{sec:reduction}

For the sake of exposition, we begin by describing our FPT reduction from NCP to MDP in a modular fashion assuming that we are given an appropriate locally dense code as advice.
Later on in \cref{sec:ldc} we give an FPT randomized algorithm to construct locally dense codes with the desired parameters and replace the advice with this construction to yield the desired FPT reduction from approximate NCP to approximate MDP with two-sided error.
We establish the following result.

\begin{theorem}\label{thm:redNCPtoMDP}
    Fix a prime power $q\geq 2$ and real numbers $\gamma,\gamma'\geq 1$ and $\alpha\in(0,1)$ additionally satisfying
    \begin{equation*}
        \gamma'\leq \frac{\gamma}{1+\alpha\gamma}.
    \end{equation*}
    Then, there is a randomized algorithm which, for $m$ large enough, on input a $\gamma$-$\NCP_q$ instance $(G,\vec{t},k)$ with $G\in\F_q^{m\times n}$, $\vec{t}\in\F_q^m$, and $k\in\Z^+$ and a $(q,\alpha,d=\gamma k,N\geq 100q^{10}\cdot (qm)^d,m',n')$-locally dense code $(A,\vec{s})$ outputs in time $\poly(m,m')$ an instance $(\Gfinal,k')$ of $\gamma'$-$\MDP_q$ with $k'<\gamma k$ satisfying the following properties with probability at least $0.99$:
    \begin{itemize}
        \item If $(G,\vec{t},k)$ is a YES instance of $\gamma$-$\NCP_q$, then $(\Gfinal, k')$ is a YES instance of $\gamma'$-$\MDP_q$;
        
        \item If $(G,\vec{t},k)$ is a NO instance of $\gamma$-$\NCP_q$, then $(\Gfinal, k')$ is a NO instance of $\gamma'$-$\MDP_q$. 
    \end{itemize}
\end{theorem}

We prove \cref{thm:redNCPtoMDP} by analyzing the following algorithm.
On input a $\gamma$-$\NCP_q$ instance $(G, \vec{t}, k)$ with $G \in \F_q^{m \times n}$ and $\vec{t} \in \F_q^m$,
we set $\C_{\textrm{int}}$ to be the code with generator matrix
\[
\Gint := \begin{pmatrix}
G & 0_{m\times n'} & -\vec{t} \\
0_{m'\times n} & A & -\vec{s}
\end{pmatrix} \in \F_q^{(m+m')\times (n+n'+1)} \text{,}
\]
where $(A,\vec{s})$ is the locally dense code described in the statement of \cref{thm:redNCPtoMDP}.
Note that $\Gint$ has full column rank (over $\F_q$) whenever $G$ and $A$ have full column rank, since we always have $\vec{s}\not\in\C(A)$ (observe that there exists at least one codeword of $\C(A)$ within distance $\alpha d <d$ of $\vec{s}$).

We will take the intersection of $\C_{\textrm{int}}$ with an appropriate random code $\C_{\textrm{rand}} \subseteq \F_q^{m+m'}$ of codimension at most $h=\lceil 7+d(1+\log_q m)\rceil$.
More precisely, we sample $\Crand$ by first sampling the entries of a parity-check matrix $H\in\F_q^{h\times (m+m')}$ independently and uniformly at random from $\F_q$ and setting $\Crand=\ker(H)$.
Then, we compute a generator matrix $\Gfinal$ of $\Cfinal=\C_{\textrm{int}} \cap \C_{\textrm{rand}}$ and
$k' := k + \alpha d$, and output $(\Gfinal, k')$ as the MDP instance.
Note that $k'\leq d/\gamma'$ by the constraints imposed on $\gamma$, $\gamma'$, and $\alpha$.

\subsection{Proof of \cref{thm:redNCPtoMDP}}

In order to prove \cref{thm:redNCPtoMDP}, we begin by establishing some useful properties of the intermediate code $\C_{\textrm{int}}$ constructed by the algorithm from \cref{sec:reduction}.

\begin{lemma}\label{lem:Cintproperties}
Fix a prime power $q\geq 2$ and real numbers $\gamma,\gamma'\geq 1$ and $\alpha\in(0,1)$ satisfying
\begin{equation*}
    \gamma'\leq \frac{\gamma}{1+\alpha \gamma}.
\end{equation*}
Given a $\gamma$-$\NCP_q$ instance $(G,\vec{t},k)$ with $G \in \F_q^{m \times n}$ and $\vec{t} \in \F_q^m$ and a $(q,\alpha,d=\gamma k,N\geq 100q^{10}\cdot (qm)^d,m',n')$-locally dense code $(A,\vec{s})$,
the algorithm from \cref{sec:reduction} constructs $\Cint=\C(\Gint)$ in time $\poly(m,m')$ satisfying the following properties:
    \begin{itemize}
        
        \item If $(G,\vec{t},k)$ is a YES instance of $\gamma$-$\NCP_q$, then there are at least $N$ nonzero vectors in $\Cint$ of Hamming weight at most $k'$. We call such vectors \emph{good};
        
        \item If $(G,\vec{t},k)$ is a NO instance of $\gamma$-$\NCP_q$, then there are at most $(qm)^d$ nonzero vectors in $\Cint$ of Hamming weight at most $d=\gamma k\geq \gamma' k'$. We call such vectors \emph{annoying}.
    \end{itemize}        
\end{lemma}
\begin{proof}
    The claim regarding the time required to construct $\Cint$ is directly verifiable.
    We proceed to argue the two items of the lemma statement.
    
    First, suppose that $(G,\vec{t},k)$ is a YES instance of $\gamma$-$\NCP_q$.
    This means that there is a vector $\vec{x}\in\F_q^n$ such that $\|G\vec{x}-\vec{t}\|_0\leq k$.
    Moreover, we know that there are at least $N\geq 100q^{10}\cdot (qm)^d$ vectors $\vec{y}\in \F_q^{n'}$ such that
    \begin{equation*}
        \|A\vec{y}-\vec{s}\|_0\leq \alpha d.
    \end{equation*}
    For each such $\vec{y}$, consider the associated vector $\vec{z}_{\vec{y}}=(\vec{x},\vec{y},1)$ and note that
    \begin{align*}
        \|\Gint \vec{z}_{\vec{y}}\|_0 = \|G\vec{x}-\vec{t}\|_0 + \|A\vec{y}-\vec{s}\|_0 \leq  k + \alpha d  = k'.
    \end{align*}
    Therefore, there are at least $N\geq 100q^{10}\cdot (qm)^d$ good vectors in $\Cint$, as desired.
    
    On the other hand, suppose that $(G,\vec{t},k)$ is a NO instance of $\gamma$-$\NCP_q$.
    This means that for every $\vec{x}\in\F_q^n$ and $\beta\in\F_q\setminus\{0\}$ it holds that $\|G\vec{x}-\beta\vec{t}\|_0>\gamma k=d$.
    Consider an arbitrary vector $\vec{z}=(\vec{x},\vec{y},-\beta)\in\F_q^{n+n'+1}$.
    We claim that if $\Gint\vec{z}$ is annoying it must be the case that $\vec{y}=\vec{0}$ and $\beta=0$.
    To see this, first note that if $\beta\neq 0$ then
    \begin{equation*}
        \|\Gint\vec{z}\|_0\geq \|G\vec{x}-\beta \vec{t}\|_0>d
    \end{equation*}
    since $(G,\vec{t},k)$ is a NO instance of $\gamma$-$\NCP_q$ and $d=\gamma k$.
    Therefore, we may assume that $\beta=0$.
    Under this assumption, it holds that
    \begin{equation*}
        \|\Gint\vec{z}\|_0\geq \|A\vec{y}\|_0>d
    \end{equation*}
    if $\vec{y}\neq \vec{0}$, which yields the claim.
    This allows us to conclude that all vectors $\vec{z}$ such that $\Gint\vec{z}$ is annoying are of the form $\vec{z}=(\vec{x},\vec{0},0)$ for some $\vec{x}\in\F_q^n$.
    As a result, the number of annoying vectors is at most
    \begin{align*}
        |\C(G)\cap\B_{q,m}(d)|&\leq |\B_{q,m}(d)|
        \leq (qm)^{d},
    \end{align*}
    as desired.
\end{proof}

We are now ready to prove \cref{thm:redNCPtoMDP} with the help of \cref{lem:Cintproperties}.
\begin{proof}[Proof of \cref{thm:redNCPtoMDP}]
    The claims regarding the time required to construct $\Cfinal$ and the bound on $k'$ are directly verifiable.
    We proceed to argue the two items of the theorem statement.
    
    Recall that we construct $\Cfinal$ by intersecting $\Cint$ with an appropriate random code $\Crand$ of codimension at most $h=\lceil 7+d(1+\log_q m)\rceil$.
    More precisely, $\Crand$ is obtained by sampling the entries of a parity-check matrix $H\in\F_q^{h\times (m+m')}$ independently and uniformly at random from $\F_q$ and setting $\Crand=\ker(H)$.
    Observe that for any given $\vec{v}\in\F_q^{m+m'}\setminus\{\vec{0}\}$ we have
    \begin{equation}\label{eq:probincode}
        \Pr_H[H\vec{v}=\vec{0}]=q^{-h}.
    \end{equation}
    Moreover, the random variables $\mathbf{1}_{\{H\vec{v}=\vec{0}\}}$ and $\mathbf{1}_{\{H\vec{w}=\vec{0}\}}$ are pairwise independent whenever $\vec{v}$ and $\vec{w}$ are linearly independent.
Let $Z_{\vec{v}}=\mathbf{1}_{\{H\vec{v}=\vec{0}\}}$ and write $Z_\cS=\sum_{\vec{v}\in \cS}Z_{\vec{v}}$ for any set $\cS$.

    Suppose that $(G,\vec{t},k)$ is a YES instance of $\gamma$-$\NCP_q$.
    By \cref{lem:Cintproperties}, this means that there are at least $100q^{10}\cdot (qm)^{d}$ good nonzero vectors in $\Cint$ of Hamming weight at most $k'$.
    Let $\cG$ denote the set of such good vectors.
    We claim that
    \begin{equation*}
        \Pr[Z_\cG=0]\leq 0.01,
    \end{equation*}
    i.e., at least one good vector survives with probability at least $0.99$ over the sampling of $\Crand$.
Note that there exists a subset $\mathcal{G}'\subseteq\mathcal{G}$ of size
\begin{equation}\label{eq:LBgprime}
    |\mathcal{G}'|\geq |\mathcal{G}|/q\geq 100q^9\cdot (qm)^d
\end{equation}
such that all vectors in $\mathcal{G}'$ are pairwise linearly independent.
This set $\cG'$ can be obtained by keeping only one element of $\mathcal{G}$ per line in $\F_q^{m+m'}$.
Note that the variables $\{Z_{\vec{v}}\}_{\vec{v}\in \mathcal{G}'}$ are pairwise independent Bernoulli random variables with success probability $q^{-h}$, and so
\cref{lem:chebyshev} guarantees that
\begin{align*}
    \Pr[Z_{\mathcal{G}}=0]\leq \Pr[Z_{\mathcal{G}'}= 0]\leq \frac{q^h}{|\mathcal{G}'|}
    \leq \frac{q^{h+1}}{|\mathcal{G}|}
    \leq 0.01,
\end{align*}
by our choice of $h$ and the lower bound on $|\mathcal{G}'|$ from \cref{eq:LBgprime}.
Therefore, with probability at least $0.99$ there is a codeword $\vec{v}\in\Cfinal\setminus\{\vec{0}\}$ such that $\|\vec{v}\|_0\leq k'$, and so $(\Gfinal,k')$ is a YES instance of $\gamma'$-$\MDP_q$.

Now, suppose that $(G,\vec{t},k)$ is a NO instance of $\gamma$-$\NCP_q$.
In this case, \cref{lem:Cintproperties} ensures that there are at most $(qm)^d$ nonzero vectors in $\Cint$ with Hamming weight at most $d$.
Let $\cA$ denote the set of such annoying vectors.
Note that
\begin{align*}
    \Pr[Z_{\mathcal{A}}\geq 1]\leq \frac{(qm)^d}{q^h}
    \leq 0.01,
\end{align*}
where the first inequality follows from \cref{eq:probincode} and a union bound over all $|\mathcal{A}|\leq (qm)^d$ annoying vectors, and the second inequality follows from the choice of $h$ above.
Therefore, all annoying vectors are removed from $\Cfinal$ with probability at least $0.99$.
This means that $\Cfinal$ has minimum distance larger than $d$, and so $(\Gfinal,k')$ is a NO instance of $\gamma'$-$\MDP_q$.
\end{proof}

\subsection{Finalizing the reduction}\label{sec:ldc}

In this section we provide a randomized construction of locally dense codes based on BCH codes~\cite{Hoc59,BC60} which can be combined with \cref{thm:redNCPtoMDP} to yield the desired FPT reduction with two-sided error and without advice.
More precisely, we prove the following theorem.
\begin{theorem}\label{thm:ldc}
    Fix a prime power $q\geq 2$ and set
    $\gamma=4q$.
    There exists a randomized algorithm which when given as input positive integers $m$ and $k\leq m$ runs in time $\poly(m)$ and outputs with probability at least $0.99$ a $(q,\alpha,d, N, m', n')$-locally dense code $(A, \vec{s})$, where
    \begin{align*}
        &m', n' \leq \poly(m),\\
        &d = \gamma k=4qk,\\
        &\alpha=1-\frac{1}{2q},\\
        &N=\frac{(qm)^{2d}}{100}\geq 100q^{10}\cdot (qm)^d,
    \end{align*}
    provided that $m$ is sufficiently large compared to $q$.
\end{theorem}

Combining \cref{thm:hardNCP,thm:redNCPtoMDP,thm:ldc} shows that $\gamma'$-$\MDP_q$ is $\W[1]$-hard under randomized reductions with two-sided error and $\gamma'=\frac{4q}{4q-1}>1$.
Then, as discussed in \cref{sec:tensoring-codes}, coupling this result with a tensoring argument immediately shows that $\gamma$-$\MDP_q$ is $\W[1]$-hard for an arbitrary constant $\gamma\geq 1$,  leading to \cref{thm:hardness-mdp}. 
Similarly, since $k'=O(k)$ in \cref{thm:redNCPtoMDP}, combining \cref{thm:FGhardnessNCP,thm:redNCPtoMDP,thm:ldc} yields \cref{thm:finegrainedMDP}.

\subsubsection{BCH codes over $\F_q$}

The following theorem states the main properties of (narrow-sense, primitive) BCH codes with design minimum distance over an arbitrary finite field (see~\cite{Gur10} for a discussion of BCH codes and related objects).
Although versions of this theorem are well-known, we present a proof in \cref{sec:bchproof} for completeness.
\begin{theorem}[$q$-ary BCH codes]\label{thm:bch}
        Fix a prime power $q$.
        Then, given integers $m'=q^r-1$ and $1\leq d\leq m'$ for some integer $r$, it is possible to construct in time $\poly(m')$ a generator matrix $\GBCH\in\F_q^{m'\times n'}$ such that $\BCH=\C(\GBCH)\subseteq\F_q^{m'}$ has minimum distance at least $d$ and codimension
    \begin{align*}
        m'-n'&\leq \lceil (d-1)(1-1/q)\rceil \log_q(m'+1).
    \end{align*}
\end{theorem}

\subsubsection{Locally dense codes from BCH codes}

We now show how to use $q$-ary BCH codes (\cref{thm:bch}) with appropriate parameters to construct locally dense codes satisfying \cref{thm:ldc}.
This construction is similar in spirit to the construction of locally dense lattices by Khot~\cite{Kho05}.

\begin{proof}[Proof of \cref{thm:ldc}]
Suppose that we are given as input $q,k,m$.
Let $d=4q k$. 
Choose $m'$ to be the smallest number of the form $q^r-1$ larger than or equal to $(d q m)^{4q}$, and set $\alpha = 1-\frac{1}{2q}$.
Let $\GBCH\in\F_q^{m'\times n'}$ be the generator matrix of the $\BCH$ code with minimum distance at least $d+1$ and codimension
\begin{equation}\label{eq:codimbch}
    m'-n'\leq \lceil d(1-1/q)\rceil \log_q(m'+1) =  d(1-1/q) \log_q(m'+1)
\end{equation}
guaranteed by \cref{thm:bch}, where the last equality holds by our choice of $d$.
We sample our locally dense code $(A,\vec{s})$ as follows:
\begin{enumerate}
    \item Set $A=\GBCH$ with $\GBCH$ as defined above;
    
    \item Sample $\vec{s}$ uniformly at random from the set of vectors in $\{0,1\}^{m'}$ of Hamming weight exactly $\alpha d$, which we denote by
    $B_{m',\alpha d}$.
\end{enumerate}

By \cref{thm:bch}, this procedure runs in time $\poly(m')=\poly(m)$.
We now show that this procedure outputs with probability at least $0.99$ an $(q, \alpha,d, N, m', n')$-locally dense code $(A, \vec{s})$, where
    \begin{align*}
        &m', n' \leq \poly(m),\\
        &d =4qk,\\
        &\alpha=1-\frac{1}{2q},\\
        &N=\frac{(qm)^{2d}}{100},
    \end{align*}
which yields \cref{thm:ldc}.
It follows directly from the length, codimension, and minimum distance of $\BCH=\C(\GBCH)=\C(A)$ that $m', n' \leq \poly(m)$ and $d = \gamma k=4qk$.
It remains to show that with probability at least $0.99$ over the sampling of $\vec{s}$ as above it holds that
\begin{equation}\label{eq:goods}
    |(\BCH-\vec{s})\cap \B_{q,m'}(\alpha d)|\geq \frac{\binom{m'}{\alpha d}}{100(m'+1)^{d(1-1/q)}} \geq \frac{(qm)^{2d}}{100} = N.
\end{equation}
We follow the reasoning of~\cite[Lemma 4.3]{Kho05} to prove the leftmost inequality.
Call a coset\footnote{A \emph{coset} of a linear code $\C\subseteq\F_q^n$ is a set of the form $\vec{v}+\C$ for some vector $\vec{v}\in\F_q^n$.} $V$ of $\BCH=\C(A)$ \emph{good} if
\begin{equation*}
    |V\cap B_{m',\alpha d}|\geq \frac{|B_{m',\alpha d}|}{100\cdot q^{m'-n'}}
\end{equation*}
and call $V$ \emph{bad} otherwise.
Note that if $V$ is a good coset it follows that
\begin{align*}
   |V\cap \B_{q,m'}(\alpha d)| \geq  |V\cap B_{m',\alpha d}|\geq \frac{|B_{m',\alpha d}|}{100\cdot q^{m'-n'}}
    \geq \frac{\binom{m'}{\alpha d}}{100(m'+1)^{d(1-1/q)}},
\end{align*}
where the third inequality holds because of \cref{eq:codimbch}.
As a result, the leftmost inequality of \cref{eq:goods} holds, and so it is enough to show that $\vec{s}$ lands in a good coset with probability at least $0.99$.

Since $\vec{s}$ is sampled uniformly at random from $B_{m',\alpha d}$, we obtain a representative of coset $V$ with probability
\begin{equation*}
    \frac{|V\cap B_{m',\alpha d}|}{|B_{m',\alpha d}|}.
\end{equation*}
Therefore, we have that
\begin{align*}
    \Pr[\text{$\vec{s}$ lands in a bad coset}] &= \sum_{V: \text{$V$ is a bad coset}}\frac{|V\cap B_{m',\alpha d}|}{|B_{m',\alpha d}|}\\
    &< \sum_{V: \text{$V$ is a bad coset}} \frac{1}{100 q^{m'-n'}}\\
    &\leq \frac{1}{100}.
\end{align*}
The first inequality follows from the definition of a bad coset.
The second inequality holds because there are at most $q^{m'-n'}$ bad cosets.

It remains to prove the rightmost inequality of \cref{eq:goods}.
Recalling that $m'\geq (d q m)^{4q}$, we have that
\begin{align*}
    \frac{\binom{m'}{\alpha d}}{100(m'+1)^{d(1-1/q)}}&\geq \frac{\binom{m'}{\alpha d}}{100(2m')^{d(1-1/q)}}\\
    &\geq \frac{(m')^{\alpha d}}{100 d^d (2m')^{d(1-1/q)}}\\
    &\geq \frac{(m')^{\frac{d}{2q}}}{100 (2d)^d}\\
    &\geq \frac{(dqm)^{2d}}{100 (2d)^d}\\
    &\geq \frac{(qm)^{2d}}{100}.
\end{align*}
The first inequality uses the fact that $m'+1\leq 2m'$. 
The second inequality holds because $\binom{m'}{\alpha d}\geq \left(\frac{m'}{\alpha d}\right)^{\alpha d}\geq \frac{(m')^{\alpha d}}{d^d}$.
The third inequality follows from the choice of $\alpha$ and since $2^{d(1-1/q)} \leq 2^d$.
The fourth inequality holds since $m'\geq (dqm)^{4q}$.
The fifth inequality holds because $d^{2d}\geq (2d)^d$ for all $d\geq 2$.
\end{proof}

\section{The FPT $\CVP_p$ to $\SVP_p$ reduction}\label{sec:CVPtoSVP}

In this section we describe and analyze an FPT reduction from approximate $\CVP_p$ to approximate $\SVP_p$ which works for all $p\geq 1$.
Our reduction is obtained by combining Khot's reduction~\cite{Kho05,BBEGSLMM21} from approximate CVP to approximate SVP with locally dense lattices stemming from Reed-Solomon-based Construction-A lattices, as first studied by Bennett and Peikert~\cite{BP22}.
In conjunction with \cref{thm:W1hardnessCVP}, our reduction yields \cref{thm:hardness-l1}, which we restate here.
\hardnesslone*

\subsection{A reduction with advice}\label{sec:latreduction}

As in \cref{sec:reduction}, we begin by describing our FPT reduction from CVP to SVP in a modular fashion assuming that we are given an appropriate locally dense lattice as advice.
In \cref{sec:ldl}, we give an FPT randomized algorithm to construct locally dense lattices with the desired parameters and replace the advice with this construction to yield the desired FPT reduction from approximate NCP to approximate MDP with two-sided error.
More precisely, we have the following result.

\begin{theorem}\label{thm:redCVPtoSVP}
    Fix real numbers $p,\gamma,\gamma'\geq 1$ and $\alpha\in (0,1)$ additionally satisfying
    \begin{align*}
        \alpha \leq \frac{((\gamma/\gamma')^p-2)^{1/p}}{\gamma}.
    \end{align*}
    Then, there is a randomized algorithm which, for $m$ large enough, on input a $\gamma$-$\CVP_p$ instance $(B,\vec{t},k)$ with $B\in\Z^{m\times n}$, $\vec{t}\in\Z^m$, and $k\in\Z^+$ and a $(p,\alpha,d=\gamma k,N\geq 10^5\cdot(2m(1+\gamma k))^{(\gamma k)^p},m',n')$-locally dense lattice $(A,\vec{s})$ outputs in time 
    polynomial in $m$, $m'$, $k$, and the description bitlength of $B$ and $\vec{t}$ an instance $(\Bfinal,k')$ of $\gamma'$-$\SVP_p$ with $k'<\gamma k$ satisfying the following properties with probability at least $0.99$:
    \begin{itemize}
        \item If $(B,\vec{t},k)$ is a YES instance of $\gamma$-$\CVP_p$, then $(\Bfinal, k')$ is a YES instance of $\gamma'$-$\SVP_p$;
        
        \item If $(B,\vec{t},k)$ is a NO instance of $\gamma$-$\CVP_p$, then $(\Bfinal, k')$ is a NO instance of $\gamma'$-$\SVP_p$. 
    \end{itemize}
\end{theorem}

We prove \cref{thm:redCVPtoSVP} via the following algorithm.
On input a $\gamma$-$\CVP_p$ instance $(B, \vec{t}, k)$ with $B \in \Z^{m \times n}$ and $\vec{t} \in \Z^m$,
we set the intermediate lattice $\lat_{\textrm{int}}$ to be the lattice generated by
\[
\Bint := \begin{pmatrix}
B & 0_{m\times n'} & -\vec{t} \\
0_{m'\times n} & A & -\vec{s}\\
\vec{0}_n & \vec{0}_{n'} & 1
\end{pmatrix} \ \text{,}
\]
where $(A,\vec{s})$ is the locally dense lattice described in \cref{thm:redCVPtoSVP}.
We add the bottom $(0,\dots,0,1)$ row to $\Bint$ to ensure that it has full column rank whenever $A$ and $B$ have full column rank as well.

Then, we add an appropriate random constraint to $\Lint$ in order to obtain the final $\SVP_p$ instance.
More precisely, let $\rho$ be a prime in the interval\footnote{Let $a=100(2m(1+\gamma k))^{(\gamma k)^p}$. By the prime number theorem, the density of primes in the interval $(a, 2a]$ is $1/\poly(\log a) =1/\poly(m,k)$.
Therefore, we can sample a prime $\rho$ with high probability in time $\poly(m,k)$ by repeatedly sampling an integer from this interval uniformly at random and then checking whether it is prime (which can be done in time $\poly(\log a)=\poly(m,k)$~\cite{AKS04}).} $(100(2m(1+\gamma k))^{(\gamma k)^p},200(2m(1+\gamma k))^{(\gamma k)^p}]$.
Sample a vector $\vec{v}\in \Z^{m+m'+1}$ by sampling each entry independently and uniformly at random from $\{0,\dots,\rho-1\}$.
We define $\Bfinal$ to be the (integral) basis of the sublattice $\Lfinal=\lat(\Bfinal)\subseteq \Lint$ defined as
\begin{equation*}
    \Lfinal =\{\vec{w}\in\Lint: \langle\vec{v},\vec{w}\rangle=0 \Mod \rho\}.
\end{equation*}
For fixed $p,\gamma\geq 1$, we can compute $\Bfinal$ given $\Bint$, $\vec{v}$, and $\rho$ as inputs in time polynomial in $m'$, $k$, $\log\rho$, and the description bitlength of $B$ and $\vec{t}$.
This can be done by first computing a basis $B'$ of the lattice $\lat_{\vec{v},\rho}=\{\vec{w}\in\Z^{m+m'+1}:\langle\vec{v},\vec{w}\rangle = 0 \Mod\rho\}$.
This step can be done in time $\poly(m,m',\log\rho)$ since $\lat_{\vec{v},\rho}$ is the Construction-A lattice $\lat_{\vec{v},\rho} = \C_{\vec{v},\rho} + \rho\Z^{m+m'+1}$ with $\C_{\vec{v},\rho}$ the $\rho$-ary linear code with parity-check matrix $H=\vec{v}^T$.
Then, we compute a basis of $\Lfinal = \Lint\cap \lat_{\vec{v},\rho}$ from $\Bint$ and $B'$ following the efficient procedure discussed in~\cite[Section 4]{Mic10}.
Overall, this requires time polynomial in $m'$, $\log\rho$, and the description bitlength of $\Bint$. In turn, this is polynomial in $m'$, $k$, and the description bitlength of $B$ and $\vec{t}$, as desired.

Finally, we compute\footnote{We can always set $\gamma$ to be an integer multiple of $\gamma'$ to ensure that $k'$ is an integer. For the sake of readability, we avoid taking floors and ceilings.} $k'=\frac{\gamma k}{\gamma'}$ and output $(\Bfinal,k')$ as the $\SVP_p$ instance.

\subsection{Proof of \cref{thm:redCVPtoSVP}}

In order to prove \cref{thm:redCVPtoSVP}, we begin by establishing some useful properties of the intermediate lattice $\lat_{\textrm{int}}$ constructed by the algorithm from \cref{sec:latreduction}.

\begin{lemma}\label{lem:Lintproperties}
    Fix real numbers $p,\gamma,\gamma'\geq 1$ and $\alpha\in (0,1)$ additionally satisfying
    \begin{align*}
        \alpha \leq \frac{((\gamma/\gamma')^p-2)^{1/p}}{\gamma}.
    \end{align*}
Given a $\gamma$-$\CVP_p$ instance $(B,\vec{t},k)$ with $B \in \Z^{m \times n}$ and $\vec{t} \in \Z^m\setminus \lat(B)$ and a $(p,\alpha,d=\gamma k,N\geq 10^5(2m(1+\gamma k))^{(\gamma k)^p},m',n')$-locally dense lattice,
the algorithm from \cref{sec:latreduction} constructs $\Lint=\lat(\Bint)$ in time polynomial in $m'$, $k$, $\log\rho$, and the description bitlength of $B$ and $\vec{t}$ satisfying the following properties with probability at least $0.99$:
    \begin{itemize}
        
        \item If $(B,\vec{t},k)$ is a YES instance of $\gamma$-$\CVP_p$, then there are at least $N$ vectors $\vec{z}$ in $\Lint$ such that $\|\vec{z}\|_p\leq k'= \frac{\gamma k}{\gamma'}$ and whose last coordinate equals $1$. 
        We call such vectors \emph{good};
        
        \item If $(B,\vec{t},k)$ is a NO instance of $\gamma$-$\CVP_p$, then there are at most $(2m(1+\gamma k))^{(\gamma k)^p}$ nonzero vectors $\vec{z}$ in $\Lint$ such that $\|\vec{z}\|_p\leq \gamma k$. We call such vectors \emph{annoying}.
    \end{itemize}        
\end{lemma}
\begin{proof}
The claim about the running time is directly verifiable.
We proceed to argue that the two items hold.

First, suppose that $(B,\vec{t},k)$ is a YES instance of $\gamma$-$\CVP_p$.
This means that there is a vector $\vec{x}\in\Z^n$ such that $\|B\vec{x}-\vec{t}\|_p\leq k$.
Moreover, we know that there are at least $N\geq 10^5\cdot(2m(1+\gamma k))^{(\gamma k)^p}$ vectors $\vec{y}\in\Z^{n'}$ such that
\begin{equation*}
    \|A\vec{y}-\vec{s}\|_p\leq \alpha d.
\end{equation*}
For each such $\vec{y}$, consider the associated vector $\vec{z}_{\vec{y}}=(\vec{x},\vec{y},1)$ and note that
\begin{align*}
    \|\Bint \vec{z}_{\vec{y}}\|_p^p = \|B\vec{x}-\vec{t}\|_p^p + \|A\vec{y}-\vec{s}\|_p^p+1
    \leq k^p + (\alpha d)^p + 1
    \leq \left(\frac{\gamma k}{\gamma'}\right)^p = (k')^p,
\end{align*}
where the last inequality follows by the constraints on $\gamma$, $\gamma'$, $\alpha$ in the lemma statement and the fact that $d=\gamma k$.
Since the last coordinate of $\Bint \vec{z}_{\vec{y}}$ is always $1$, 
we conclude that there are at least $N$ good vectors.

On the other hand, suppose that $(B,\vec{t},k)$ is a NO instance of $\gamma$-$\CVP_p$.
This means that for every vector $\vec{x}\in\Z^n$ and $\alpha\in\Z\setminus\{0\}$ it holds that $\|B\vec{x}-\alpha\vec{t}\|_p>\gamma k$.
Consider any vector $\vec{z}=(\vec{x},\vec{u},\alpha)$.
Note that every annoying vector $\vec{z}$ must have $\alpha=0$ by this property.
Furthermore, because $\alpha=0$, it must also be the case that $\vec{u}=0$, since $\|A\vec{u}\|_p>\gamma k$ for all $\vec{u}\in\Z^{n'}\setminus\{\vec{0}\}$.
Therefore, the number of annoying vectors is upper bounded by
\begin{equation*}
    |\lat(B)\cap\B^{(p)}_m(\gamma k)| \leq \binom{m}{(\gamma k)^p} (1+2\gamma k)^{(\gamma k)^p}\leq (2m(1+\gamma k))^{(\gamma k)^p}.\qedhere
\end{equation*}
\end{proof}

We are now ready to prove \cref{thm:redCVPtoSVP} with the help of \cref{lem:Lintproperties}.
\begin{proof}[Proof of \cref{thm:redCVPtoSVP}]
    The claims regarding the time required to construct $\Lfinal$ and the bound on $k'$ are directly verifiable.
    We proceed to argue the two items of the theorem statement.

Recall that $\Lfinal$ is defined as the random sublattice
\begin{equation*}
    \Lfinal =\{\vec{w}\in\Lint: \langle\vec{v}, \vec{w}\rangle=0 \Mod \rho\},
\end{equation*}
where $\rho$ is a prime in the interval $(100(2m(1+\gamma k))^{(\gamma k)^p}, 200(2m(1+\gamma k))^{(\gamma k)^p}]$ and the entries of $\vec{v}\in \Z^{m+m'+1}$ are sampled independently and uniformly at random from $\{0,\dots,\rho-1\}$.

Suppose that $(B,\vec{t},k)$ is a YES instance of $\gamma$-$\CVP_p$.
Our goal is to show that, with probability at least $0.99$ over the randomness of the algorithm, there is $\vec{w'}\in\Lfinal$ such that $\|\vec{w'}\|_p\leq k'=\frac{\gamma k}{\gamma'}$.
In this case, \cref{lem:Lintproperties} ensures that there are at least $N=10^5\cdot(2m(1+\gamma k))^{(\gamma k)^p}\geq 100\rho$ good vectors $\vec{w}\in\Lint$ such that $\|\vec{w}\|_p\leq k'$ and whose last coordinate equals $1$.
Then, it follows that any two distinct good vectors $\vec{w}^{(1)}, \vec{w}^{(2)}\in\Lint$ are linearly independent modulo $\rho$, because their last coordinates equal $1$ and all their entries are bounded in absolute value by $k' < \rho/2$.
As a result, the random variables $\mathbf{1}_{\{\langle\vec{v}, \vec{w}^{(1)}\rangle=0\Mod \rho\}}$ and $\mathbf{1}_{\{\langle\vec{v}, \vec{w}^{(2)}\rangle=0\Mod \rho\}}$ are independent.
Since all these (at least $N$) pairwise independent random variables follow a Bernoulli distribution with success probability exactly $1/\rho$, \cref{lem:chebyshev} implies that the probability over the sampling of $\vec{v}$ that there exists at least one good vector $\vec{w}\in\Lint$ such that $\langle\vec{v},\vec{w}\rangle=0\Mod \rho$, and hence $\vec{w}\in\Lfinal$, is at least
\begin{equation*}
    1 - \rho/N \geq 0.99.
\end{equation*}
It follows that $(\Bfinal,k')$ is a YES instance of $\gamma'$-$\SVP_p$ with probability at least $0.99$.

Now, suppose that $(B,\vec{t},k)$ is a NO instance of $\gamma$-$\CVP_p$.
\cref{lem:Lintproperties} guarantees that there are at most $N'=(2m(1+\gamma k))^{(\gamma k)^p}$ annoying nonzero vectors $\vec{w}\in\Lint$ such that $\|\vec{w}\|_p\leq \gamma k$.
For any given nonzero integer vector $\vec{w}$, the probability (over the sampling of $\vec{v}$) that $\langle\vec{v}, \vec{w}\rangle=0\Mod \rho$ is exactly $1/\rho$.
Consequently, a union bound over the at most $N'$ annoying vectors in $\Lint$ shows that the probability that there exists a vector $\vec{w'}\in\Lfinal$ such that $\|\vec{w'}\|_p\leq \gamma k$ is at most
\begin{equation*}
    N'/\rho \leq 0.01.
\end{equation*}
It follows that $(\Bfinal,k')$ is a NO instance of $\gamma'$-$\SVP_p$ with probability at least $0.99$.
\end{proof}

\subsection{Finalizing the reduction}\label{sec:ldl}

In this section we provide a randomized construction of locally dense lattices based on Construction~A lattices stemming from Reed-Solomon codes, which can be combined with \cref{thm:redCVPtoSVP} to yield the desired randomized FPT reduction with two-sided error and without advice.
More precisely, we prove the following theorem.
\begin{theorem}\label{thm:ldl}
    Fix real numbers $p\geq 1$ and $\gamma'\in[1,2^{1/p})$.
    Let $\eps=(\gamma')^{-p}-1/2>0$ and set\footnote{We made no effort to optimize constants.}
    \begin{equation*}
        \gamma=\left\lceil\max\left(12/\eps,\frac{1}{(1+\eps/2)^{1/p}-1}\right)\right\rceil.
    \end{equation*}
    There exists a randomized algorithm which when given as input positive integers $m$ and $k$ runs in time $\poly(m,k)$ and outputs with probability at least $0.99$ a $(p,\alpha,d, N, m', n')$-locally dense lattice $(A, \vec{s})$, where
\begin{align*}
    &\alpha=\left(\frac{1}{(\gamma')^p}-\frac{2}{\gamma^p}\right)^{1/p}, \\
    &d= \gamma k,\\
    &N= (2m(1+\gamma k))^{3(\gamma k)^p}\geq 10^5\cdot (2m(1+\gamma k))^{(\gamma k)^p},\\
    &n',m'=\poly(m,k),
\end{align*}
provided that $m$ is sufficiently large compared to $p$.
\end{theorem}

Combining \cref{thm:W1hardnessCVP,thm:redCVPtoSVP,thm:ldl} yields \cref{thm:hardness-l1}.
Similarly, observing that $k'=O(k)$ in \cref{thm:redCVPtoSVP}, combining \cref{thm:FGhardnessCVP,thm:redCVPtoSVP,thm:ldl} yields \cref{thm:finegrainedSVP}.

\subsubsection{Locally dense lattices from Reed-Solomon codes}

The locally dense lattices described in \cref{thm:ldl} are obtained via Construction A lattices based on Reed-Solomon codes, which were analyzed in~\cite{BP22}.

For a fixed prime $q$, we define the Reed-Solomon code with block length $q$ and dimension $\ell$, which we denote by $\RSparam$, as
\begin{equation*}
    \RSparam=\left\{(f(\zeta))_{\zeta\in\F_q}:f\in\F_q[x],\deg(f)<\ell\right\}.
\end{equation*}
Note that $\ell\leq q$.
It is well known that it is possible to construct a generator matrix for $\RSparam$ in time $\poly(q)$.

Given a Reed-Solomon code $\RSparam$, we define the associated Reed-Solomon (Construction A) lattice $\latRSparam$ as
\begin{equation*}
    \latRSparam=\{\vec{x}\in\Z^q: \vec{x}\Mod q\in\RSparam\}=\RSparam+q\Z^q.
\end{equation*}
Since we can construct a generator matrix of $\RSparam$ in time $\poly(q)$, we can also construct a basis of $\latRSparam$ in time $\poly(q)$.
Bennett and Peikert established two important properties of $\latRS$.

\begin{lemma}[\protect{\cite[Theorem 3.1, adapted]{BP22}}]\label{lem:distlatRS}
    Suppose that $\ell\leq q/2$.
    Then, it holds that
    \begin{equation*}
      \lambda_1^{(p)}(\latRSparam)\geq (2\ell)^{1/p}
    \end{equation*}
    for every $p\geq 1$.
\end{lemma}

\begin{lemma}[\protect{\cite[Lemma 3.3, adapted]{BP22}}]\label{lem:shiftlatRS}
    Fix $\delta>0$ and denote the set of vectors in $\{0,1\}^q$ of Hamming weight $w$ by $B_{q,w}$.
    Sample $\vec{s}$ uniformly at random from $B_{q,w}$.
    Then,
    \begin{equation*}
      \Pr\left[\left|(\latRSparam-\vec{s})\cap B_{q,w}\right|\geq \delta\cdot \binom{q}{w}/q^\ell\right]\geq 1-\delta.
    \end{equation*}
\end{lemma}

We now move to prove \cref{thm:ldl}.
\begin{proof}[Proof of \cref{thm:ldl}]
Fix real numbers $p\geq 1$ and $\gamma'\in[1,2^{1/p})$, and set $\eps=(\gamma')^{-p}-1/2$ and $\gamma=\left\lceil\max\left(12/\eps,\frac{1}{(1+\eps/2)^{1/p}-1}\right)\right\rceil$.
Suppose that we are given $k$ and $m$ as inputs.
Choose integers
\begin{align*}
    &\ell = \left\lceil \frac{(1+\gamma k)^p}{2}\right\rceil,\\
    & w=\left\lfloor((\gamma/\gamma')^p-2)k^p\right\rfloor,
\end{align*}
and $q$ to be the smallest prime larger than
\begin{equation*}
    (300w(\gamma k)^p\cdot (2m(1+\gamma k)))^{9/\eps}=\poly(m,k).
\end{equation*}
Note that $q=\poly(m,k)$ by Bertrand's postulate and that we can naively verify whether $q$ is prime in time $\poly(\log m,\log k)$~\cite{AKS04}.
We consider the following candidate construction of a locally dense lattice $(A,\vec{s})$, which runs in overall time $\poly(q)=\poly(m,k)$:
\begin{enumerate}
    \item Set $A$ to be a basis of $\latRSparam$;
    
    \item Sample $\vec{s}$ uniformly at random from $B_{q,w}$. (Recall that $B_{q,w}$ denotes the set of vectors in $\{0,1\}^q$ of Hamming weight $w$.)
\end{enumerate}
We now argue that $(A,\vec{s})$ is a $(p,\alpha,d,N,m',n')$-locally dense lattice with the properties described in \cref{thm:ldl} with probability at least $0.99$ over the sampling of $(A,\vec{s})$.
First, note that $m'=q$ and $n'=\ell$, and so it can be directly verified that $n',m'=\poly(m,k)$.
To bound $d$, note that $\ell\leq q/2$.
Therefore, \cref{lem:distlatRS} and the choice of $\ell$ above guarantee that
\begin{equation*}
    \lambda^{(p)}_1(\lat(A))\geq (2\ell)^{1/p}>\gamma k,
\end{equation*}
and so we may take $d= \gamma k$.
It remains to show that
\begin{equation*}
    |(\lat(A)-\vec{s})\cap \B^{(p)}_q(\alpha d)|\geq N
\end{equation*}
with probability at least $0.99$ over the sampling of $\vec{s}$, where
\begin{align*}
    &\alpha=\frac{((\gamma/\gamma')^p-2)^{1/p}}{\gamma},\\
    &N= (2m(1+\gamma k))^{3(\gamma k)^p}.
\end{align*}
Since $w\leq (\alpha d)^p$ and $B_{q,w}\subseteq \B^{(p)}_q(\alpha d)$, it suffices to show that
\begin{equation*}
    |(\lat(A)-\vec{s})\cap B_{q,w}|\geq N
\end{equation*}
with probability at least $0.99$ over the sampling of $\vec{s}$.
Invoking \cref{lem:shiftlatRS} with $\delta=0.01$ shows that
\begin{equation*}
    \Pr\left[|(\lat(A)-\vec{s})\cap B_{q,w}|\geq \frac{\binom{q}{w}}{100 q^\ell}\right]\geq 0.99.
\end{equation*}
We claim that $\frac{\binom{q}{w}}{100 q^\ell}\geq N$ with the choices of $\ell$, $w$, and $q$ above, which concludes the proof of \cref{thm:ldl}.
To see this, first note that
\begin{align}
    w-\ell &= \lfloor ((\gamma/\gamma')^p-2)k^p \rfloor - \left\lceil \frac{(1+\gamma k)^p}{2}\right\rceil \nonumber \\
    &\geq  ((\gamma/\gamma')^p-2)k^p - \frac{(1+\gamma k)^p}{2} - 2\nonumber\\
    &\geq ((1/2+\eps)\gamma^p-2)k^p - (1/2+\eps/2)(\gamma k)^p - 2 \label{eq:boring1}\\
    &= \frac{\eps}{2}(\gamma k)^p - 2k^p - 2\nonumber\\
    &\geq \frac{\eps}{2}(\gamma k)^p - 4k^p \label{eq:newboring}\\
    &\geq \frac{\eps}{3}(\gamma k)^p,\label{eq:boring2}
\end{align}
where \cref{eq:boring1} uses the observation that $(1+\gamma k)^p \leq (1+\eps/2)(\gamma k)^p$,
which follows from the fact that $\gamma\geq\frac{1}{(1+\eps/2)^{1/p}-1}$, \cref{eq:newboring} holds because $k\geq 1$, and
\cref{eq:boring2} uses the fact that $\gamma\geq 12/\eps$.

Then, we have that
\begin{align}
    \frac{\binom{q}{w}}{100 q^\ell}&\geq \frac{q^{w-\ell}}{100 w^w}\label{eq:goodvec1}\\
    &\geq \frac{q^{\frac{\eps}{3}(\gamma k)^p}}{100 w^w}\label{eq:newgoodvec}\\
    &\geq \frac{(300w(2m(1+\gamma k)))^{3(\gamma k)^p}}{100 w^w}\label{eq:goodvec2}\\
    &\geq (2m(1+\gamma k))^{3(\gamma k)^p}\label{eq:goodvec3}\\
    &=N.\nonumber
\end{align}
\cref{eq:goodvec1} holds because $\binom{q}{w}\geq \left(q/w\right)^w$.
\cref{eq:newgoodvec} follows from the lower bound on $w-\ell$ from \cref{eq:boring2}.
\cref{eq:goodvec2} holds since $q > (300w(\gamma k)^p\cdot (2m(1+\gamma k)))^{9/\eps}$.
Finally, \cref{eq:goodvec3} holds because $w<3(\gamma k)^p$.
\end{proof}

\section{$\W[1]$-hardness of $\SVP_p$ for any approximation factor}\label{sec:NCPtoSVPtensor}

In this section we analyze an FPT reduction from approximate $\NCP_2$ to approximate $\SVP_p$ which when combined with results of Haviv and Regev~\cite{HR12} leads to \cref{thm:hardness-lp-tensor}, which we restate here.

\hardnesslp*

\subsection{The Haviv-Regev conditions for tensoring of $\SVP$ instances}

We will use the following generalization of a result of Haviv and Regev~\cite{HR12} which establishes conditions under which an SVP instance behaves well under tensoring.
\begin{restatable}{lemma}{havivregev}\label{lem:havivregev}
    Fix an integer $c\geq 1$ and real numbers $p,\gamma\geq 1$.
Suppose that $(B,k)$ with $B\in\Z^{m\times n}$ and $k\in\Z^+$ is an instance of $\gamma$-$\SVP_p$ with the additional property that if $(B,k)$ is a NO instance of $\gamma$-$\SVP_p$, then every nonzero vector $\vec{w}\in\lat(B)$ satisfies at least one of the following conditions, where $d=\gamma k$:
\begin{itemize}
    \item $\|\vec{w}\|_0> d^p$;
    
    \item $\vec{w}\in 2\Z^m$ and $\|\vec{w}\|_0> d^p/2^p$;
    
    \item $\vec{w}\in 2\Z^m$ and $\|\vec{w}\|_p>d^{c+3p/2}$.
\end{itemize}
Then, $(B^{\otimes c},k^c)$ is a YES (resp.\ NO) instance of $\gamma^c$-$\SVP_p$ if $(B,k)$ is a YES (resp.\ NO) instance of $\gamma$-$\SVP_p$, where $B^{\otimes c}$ denotes the $c$-fold tensor product of $B$ with itself.
\end{restatable}

Haviv and Regev~\cite[Lemma 3.4]{HR12} proved a similar result for the case $p=2$ only. It turns out to be not hard to generalize for any $p\geq 1$.
Note, however, that the conditions of \cref{lem:havivregev} depend a priori on the (constant) number of times $c$ we will be tensoring the base $\SVP_p$ instance.
We provide a proof of \cref{lem:havivregev} in \cref{sec:HRgen}.

\subsection{The FPT $\NCP_2$ to $\SVP_p$ reduction amenable to tensoring}\label{sec:redHR}
We proceed to describe an FPT reduction from approximate $\NCP_2$, which we know is $\W[1]$-hard by \cref{thm:hardNCP}, to approximate $\SVP_p$ that yields the following result.
\begin{theorem}\label{thm:redNCPtoSVP}
    Fix an even integer $\gamma\geq 2$, an integer $c\geq 1$, and real numbers $p>1$, $\gamma'\geq 1$, and $\alpha\in (1/2+2^{-p},1)$ additionally satisfying
    \begin{equation*}
        \gamma'< \left(\frac{\gamma}{2^p+1+\alpha \gamma}\right)^{1/p}.
    \end{equation*}
    Then, there is a randomized algorithm which, for $m$ large enough, on input a $\gamma$-$\NCP_2$ instance $(G,\vec{t},k)$ with $G\in\F_2^{m\times n}$, $\vec{t}\in\F_2^m$, and $k\in\Z^+$ outputs in time $\poly(m)$ an instance $(\Bfinal,k')$ of $\gamma'$-$\SVP_p$ with $k'=(2^p k+\alpha \gamma k+1)^{1/p}<(\gamma k)^{1/p}$ satisfying the following properties with probability at least $0.9$:
    \begin{itemize}
        \item If $(G,\vec{t},k)$ is a YES instance of $\gamma$-$\NCP_2$, then $(\Bfinal, k')$ is a YES instance of $\gamma'$-$\SVP_p$;
        
        \item If $(G,\vec{t},k)$ is a NO instance of $\gamma$-$\NCP_2$, then $(\Bfinal, k')$ is a NO instance of $\gamma'$-$\SVP_p$ such that every nonzero vector $\vec{w}\in\lat(\Bfinal)$ satisfies at least one of the following conditions:
\begin{itemize}
    \item $\|\vec{w}\|_0> (\gamma' k')^p$;
    
    \item $\vec{w}\in 2\Z^m$ and $\|\vec{w}\|_0> (\gamma' k')^p/2^p$;
    
    \item $\vec{w}\in 2\Z^m$ and $\|\vec{w}\|_p>(\gamma' k')^{c+3p/2}$.
\end{itemize}
    \end{itemize}
\end{theorem}

Combining \cref{thm:redNCPtoSVP} with \cref{lem:havivregev} and \cref{thm:hardNCP} immediately yields \cref{thm:hardness-lp-tensor}.
This is because \cref{lem:havivregev} guarantees that we can directly tensor the $\gamma'$-$\SVP_p$ instances from \cref{thm:redNCPtoMDP} with $\gamma'>1$ an arbitrary (constant) number of times $c\geq 1$ to conclude that $\gamma''$-$\SVP_p$ is $\W[1]$-hard for any constant $\gamma''\geq 1$.

We proceed to describe the algorithm we use to prove \cref{thm:redNCPtoSVP}.
First, we need to set up some auxiliary objects and lemmas.
Suppose that we are given as input an instance $(G,\vec{t},k)$ of $\gamma$-$\NCP_2$, where $G\in\F_2^{m\times n}$, $\vec{t}\in\F_2^m$, and $k\in\Z^+$.
Let $d=\gamma k$.
We denote by $\BNCP\in\Z^{m\times m}$ the basis of the Construction A lattice
\begin{equation*}
    \latNCP = \C(G)+2\Z^m,
\end{equation*}
which we can compute in time $\poly(m)$.

With some hindsight, set $m'$ to be the smallest integer of the form $2^r-1$ larger than
\begin{equation*}
    \max\left(m+1, (10^8  d^{12c})^{\frac{1}{\alpha -(1/2+2^{-p})}}\right) =\poly(m),
\end{equation*}
and let $\BCH\subseteq\F_2^{m'}$ be a binary BCH code with minimum distance at least $d+1$ and codimension
\begin{equation*}
    h\leq \left\lceil \frac{d}{2}\right\rceil \log(m'+1) =  \frac{d}{2} \log(m'+1)
\end{equation*}
guaranteed by \cref{thm:bch}.
We denote by $\BBCH\in\Z^{m'\times m'}$ the basis of the Construction A lattice
\begin{equation*}
    \latBCH=\BCH+2\Z^{m'}.
\end{equation*}
Note that we can compute a basis of $\latBCH$ in time $\poly(m')=\poly(m)$.
Furthermore, we sample a target vector $\vec{s}\in\F_2^{m'}$ uniformly at random from $B_{m',\alpha d}$, where we recall that $B_{m',\alpha d}$ is the set of vectors in $\{0,1\}^{m'}$ with Hamming weight $\alpha d$.
As in previous sections, the tuple $(\BBCH,\vec{s})$ satisfies a local density property with high probability over the sampling of $\vec{s}$, as described in the following lemma.
\begin{lemma}\label{lem:propBCHlat}
    It holds with probability at least $0.99$ over the sampling of $\vec{s}$ from $B_{m', \alpha d}$
    that
    \begin{equation*}
        |(\latBCH-\vec{s})\cap \B^{(p)}_{m'}((\alpha d)^{1/p})|\geq \frac{\binom{m'}{\alpha d}}{100(m'+1)^{d/2}}=: N.
    \end{equation*}
\end{lemma}
\begin{proof}
Consider the vectors $\vec{v}$ of $\latBCH$ which lie in $\BCH$ (seen as a subset of $\Z^{m'}$).
Since both $\vec{v}$ and $\vec{s}$ lie in $\{0,1\}^{m'}$, it follows that
\begin{equation*}
   \|\vec{v} - \vec{s}\|_0 = \|\vec{v}-\vec{s}\Mod 2\|_0= \|\vec{v}-\vec{s}\|^p_p.
\end{equation*}
This implies that
\begin{equation}\label{eq:0top}
    |(\latBCH-\vec{s})\cap \B^{(p)}_{m'}((\alpha d)^{1/p})| \geq |(\BCH-\vec{s}\Mod 2)\cap \B_{2,m'}(\alpha d)|.
\end{equation}
Following the derivation of \cref{eq:goods} from the proof of \cref{thm:ldc} on locally dense codes with $q=2$ guarantees that with probability at least $0.99$ over the sampling of $\vec{s}$ it holds that $|(\BCH-\vec{s}\Mod 2)\cap \B_{2,m'}(\alpha d)|\geq\frac{\binom{m'}{\alpha d}}{100(m'+1)^{d/2}}$.
Combining this with \cref{eq:0top} yields the desired result.
\end{proof}

Equipped with the above, we consider the intermediate lattice $\Lint$ generated by the basis
\begin{equation} \label{eq:Bint-tensor-svp}
    \Bint := \begin{pmatrix}
2\BNCP & 0_{m\times m'} & -2\vec{t} \\
0_{m'\times m} & \BBCH & -\vec{s}\\
\vec{0}_m & \vec{0}_{m'} & 1
\end{pmatrix} \in\Z^{(m+m'+1)\times (m+m'+1)} \ \text{.}
\end{equation}
The bottom $(0,\dots,0,1)$ row is added to ensure that $\Bint$ has full column rank over $\R$.

Then, we add a random constraint $\Lint$ to obtain our final SVP instance.
More precisely, we set $\rho$ to be a prime in the interval $(N/100, N/50]$ (as discussed in \cref{sec:CVPtoSVP}, we can sample $
\rho$ with high probability in time $\poly(m)$),
sample a vector $\vec{v}\in\Z^{m+m'+1}$ by sampling each entry of $\vec{v}$ independently and uniformly at random from $\{0,\dots,\rho-1\}$, and construct in time $\poly(m)$ a basis $\Bfinal$ of the random sublattice
\begin{equation*}
    \Lfinal=\{\vec{w}\in\Lint:\langle\vec{v},\vec{w}\rangle=0\Mod\rho\}
\end{equation*}
by following the procedure discussed in \cref{sec:latreduction}.
Then, we set
$k'=(2^p k+\alpha d+1)^{1/p}$ 
and output $(\Bfinal,k')$ as our $\gamma'$-$\SVP_p$ instance.\footnote{Our choice of $k'$ may not be an integer. For the sake of readability, we avoid working through the argument with floors and ceilings. It is also relevant to note that $\gamma$ can be chosen so that our choice of $k'$ is the $p$th root of an integer, which already matches the requirements of the definition of approximate $\SVP_p$ in~\cite{BBEGSLMM21}.}

\subsection{Proof of \cref{thm:redNCPtoSVP}}

In order to prove \cref{thm:redNCPtoSVP}, we begin by establishing some useful properties of the intermediate lattice $\Lint$ captured in the following lemma.

\begin{lemma}\label{lem:LintpropertiesHR}
    Fix an even integer $\gamma\geq 2$, an integer $c\geq 1$, and real numbers $p>1$, $\gamma'\geq 1$, and $\alpha\in (1/2+2^{-p},1)$ additionally satisfying
    \begin{equation*}
        \gamma'< \left(\frac{\gamma}{2^p+1+\alpha \gamma}\right)^{1/p}.
    \end{equation*}
Given a $\gamma$-$\NCP_2$ instance $(G,\vec{t},k)$ with $G \in \F_2^{m \times n}$, $\vec{t} \in \Z^m$, and $k\in\Z^+$,
the algorithm from \cref{sec:redHR} constructs $\Lint=\lat(\Bint)\subseteq \Z^{m+m'+1}$ in time $\poly(m)$ satisfying the following properties with probability at least $0.99$, where we recall that $d=\gamma k$ and $k'=(2^p k+\alpha d+1)^{1/p}$:
    \begin{itemize}
        
        \item If $(G,\vec{t},k)$ is a YES instance of $\gamma$-$\NCP_2$, then there are at least $N= \frac{\binom{m'}{\alpha d}}{100(m'+1)^{d/2}}$ vectors $\vec{w}$ in $\Lint$ such that $\|\vec{w}\|_p\leq k'$ and whose last coordinate equals $1$. We call such vectors \emph{good};
        
        \item If $(G,\vec{t},k)$ is a NO instance of $\gamma$-$\NCP_2$, then there are at most $A\leq 10^{-5} N$ nonzero vectors $\vec{w}$ in $\Lint$ that satisfy \emph{all} of the following properties:
        \begin{itemize}
            \item $\|\vec{w}\|_0\leq (\gamma' k')^p$;
            
            \item Either $\vec{w}\not\in 2\Z^{m+m'+1}$ or $\|\vec{w}\|_0\leq (\gamma' k')^p/2^p$;
            
            \item Either $\vec{w}\not\in 2\Z^{m+m'+1}$ or $\|\vec{w}\|_p\leq (\gamma' k')^{c+3p/2}$.
        \end{itemize}
        We call such vectors \emph{annoying}.\footnote{Annoying vectors are the ones that do not satisfy the properties of NO instances laid out in \cref{thm:redNCPtoSVP}.}
    \end{itemize}        
\end{lemma}
\begin{proof}
The claim about the running time of the algorithm is directly verifiable.
We proceed to argue the two items in the lemma statement.

Suppose that $(G,\vec{t},k)$ is a YES instance of $\gamma$-$\NCP_2$.
This means that there is a codeword $\vec{c}\in\C(G)$ such that
\begin{equation*}
  \|\vec{c}-\vec{t}\|_0\leq k.  
\end{equation*}
Noting that $\C(G)\subseteq\latNCP$ (when seen as a subset of $\{0,1\}^m\subseteq \Z^m$), we conclude that there is $\vec{x}\in\Z^m$ such that $\BNCP \vec{x}=\vec{c}$, and so
\begin{equation*}
    \|\BNCP \vec{x}-\vec{t}\|_p^p \leq k.
\end{equation*}
Moreover, by \cref{lem:propBCHlat} we also know that with probability at least $0.99$ there are at least $N=  \frac{\binom{m'}{\alpha d}}{100(m'+1)^{d/2}}$ vectors $\vec{y}\in\Z^{m'}$ such that
\begin{equation*}
    \|\BBCH \vec{y}-\vec{s}\|_p^p\leq \alpha d.
\end{equation*}
For each such good $\vec{y}$, consider the vector $\vec{z}_{\vec{y}}=(\vec{x},\vec{y},1)$.
Then, we have that
\begin{equation*}
    \|\Bint \vec{z}_{\vec{y}}\|_p^p = \|2\BNCP \vec{x}-2\vec{t}\|_p^p + \|\BBCH \vec{y}-\vec{s}\|_p^p +1 \leq 2^p k +\alpha d+1,
\end{equation*}
and so $\|\Bint \vec{z}_{\vec{y}}\|_p\leq (2^p k+\alpha d+1)^{1/p}= k'$.
Furthermore, the last coordinate of $\Bint \vec{z}_{\vec{y}}$ is always $1$.
As a result, there are at least $N$ good vectors in $\Lint$.

On the other hand, suppose that $(G,\vec{t},k)$ is a NO instance of $\gamma$-$\NCP_2$.
This means that for every $\vec{c}\in\C(G)$ it holds that
\begin{equation*}
    \|\vec{c}-\vec{t}\|_0>d=\gamma k\geq (\gamma' k')^p,
\end{equation*}
where the last inequality follows by our choice of $k'$ and the constraints on $p$,  $\gamma$, $\gamma'$, and $\alpha$.
Recall that our goal is to bound the number of annoying vectors in $\Lint$ appropriately.
Consider an arbitrary vector $\vec{z}=(\vec{x},\vec{y},\beta)\in\Z^{m+m'+1}$.
We proceed by case analysis:
\begin{enumerate}
    \item $\beta\not\in 2\Z$:
    In this case, we have
    \begin{align*}
        \|\Bint \vec{z}\|_0 \geq \|\BNCP\vec{x}-\beta\vec{t}\|_0 \geq \|G\vec{x}-\vec{t}\Mod 2\|_0>d\geq (\gamma' k')^p,
    \end{align*}
    and so no vector of this form is annoying.
    
    \item $\beta\in 2\Z$ and $\BBCH\vec{y}\not\in 2\Z^{m'}$:
    In this case, we have
    \begin{equation*}
        \|\Bint\vec{z}\|_0 \geq \|\BBCH\vec{y}-\beta\vec{s}\|_0 \geq \|\BBCH\vec{y}\Mod 2\|_0>d\geq (\gamma' k')^p,
    \end{equation*}
    and so no vector of this form is annoying.
    The third inequality uses the fact that $\BBCH\vec{y}\Mod 2$ is a nonzero codeword of $\BCH$, which has minimum distance larger than $d=\gamma k$.
    
    \item $\beta\in 2\Z$ and $\BBCH\vec{y}\in 2\Z^{m'}$:
    In this case, it holds that all coordinates of $\Bint\vec{z}$ are even.
    Therefore, in order for $\Bint\vec{z}$ to be annoying it must be that $\|\Bint\vec{z}\|_0\leq (\gamma' k')^p/2^p\leq d/2^p$ and $\|\Bint\vec{z}\|_p\leq (\gamma' k')^{c+3p/2}\leq d^{3c}$.
    There are at most
    \begin{equation*}
        (2d^{3c}+1)^{d/2^p}\binom{m+m'+1}{d/2^p}
    \end{equation*}
    vectors in $\Lint$ with these properties.
\end{enumerate}
We conclude that there are at most $A=(2d^{3c}+1)^{d/2^p}\binom{m+m'+1}{d/2^p}$ annoying vectors in $\Lint$.
Finally, we claim that, since we chose $m'$ to be larger than $\max\left(m+1, (10^8 d^{12c})^{\frac{1}{\alpha -(1/2+2^{-p})}}\right)$ in \cref{sec:redHR}, it follows that $A\leq 10^{-5} N$.
To see this, note that
\begin{equation*}
    A\leq (2d^{3c}+1)^{d/2^p}(m+m'+1)^{d/2^p}\leq (3d^{3c})^{d}(2m')^{d/2^p}
\end{equation*}
and
\begin{equation*}
    N\geq \frac{(m')^{(\alpha-1/2)d}}{100\cdot 2^{d}\cdot d^d },
\end{equation*}
where we have used the fact that $m'\geq m+1$.
Therefore, after simple algebraic manipulation, it follows that the desired inequality holds whenever
\begin{equation*}
    (m')^{(\alpha-(1/2+2^{-p}))d}\geq 10^7 \cdot (6 d^{4c})^d,
\end{equation*}
which is in turn satisfied by our choice of $m'$.
\end{proof}

\begin{proof}[Proof of \cref{thm:redNCPtoSVP}]
\cref{thm:redNCPtoSVP} follows by combining \cref{lem:LintpropertiesHR} with a standard sparsification argument, as carried out in the proof of \cref{thm:redCVPtoSVP}.
To avoid repetition, we omit the full argument here and simply give a sketch.

Let $\Lint := \lat(\Bint)$, where $\Bint$ is as defined in \cref{eq:Bint-tensor-svp}.
First, we note that the random sublattice $\Lfinal$ of $\Lint$ is defined in \cref{sec:redHR} with respect to a prime $\rho$ that satisfies $100A\leq\rho\leq N/100$.
Additionally, we note that there at least $N$ distinct good vectors of the form $(\vec{x}^T, \vec{y}_i^T, 1)^T \in \Lint$ such that $\vec{y}_i \in (\lat(\BBCH) - \vec{s})$ is a binary vector and $\vec{y}_i, \vec{y}_j$ have $1$s in distinct coordinates for all $i, j \in [N], i \neq j$.
Then, by \cref{lem:chebyshev} and the fact that any two such distinct good vectors $(\vec{x}^T, \vec{y}_i^T, 1)^T, (\vec{x}^T, \vec{y}_j^T, 1)^T$ are linearly independent modulo $\rho$ (which follows from the fact that $\vec{y}_i, \vec{y}_j$ have $1$s in distinct coordinates), with probability at least $0.99$ in the YES case there is at least one good vector left in $\Lfinal$, in which case $(\Bfinal,k')$ is a YES instance of $\gamma'$-$\SVP_p$.
Moreover, using the fact that each annoying vector will be kept with probability at most $1/\rho$ and taking a union bound, with probability at least $0.99$ in the NO case there are no annoying vectors left in $\Lfinal$, meaning that $(\Bfinal,k')$ is a NO instance of $\gamma'$-$\SVP_p$ with the additional properties outlined in \cref{thm:redNCPtoSVP}.
\end{proof}

\bibliographystyle{alpha}
\bibliography{refs}

\appendix

\section{Proof of \cref{thm:bch}}\label{sec:bchproof}

We restate \cref{thm:bch} here for convenience.
\begin{theorem}[Restatement of \cref{thm:bch}, $q$-ary BCH codes]
        Fix a prime power $q$.
        Then, given integers $m'=q^r-1$ and $1\leq d\leq m'$ for some positive integer $r$, it is possible to construct in time $\poly(m')$ a generator matrix $\GBCH\in\F_q^{m'\times n'}$ such that $\BCH=\C(\GBCH)\subseteq\F_q^{m'}$ has minimum distance at least $d$ and codimension
    \begin{align*}
        m'-n'&\leq \lceil (d-1)(1-1/q)\rceil \log_q(m'+1).
    \end{align*}
\end{theorem}
\begin{proof}
Let $\alpha$ be a primitive element of $\F_{q^r}$.
For a given block length $m'=q^r-1$ and design distance $d\leq m'$, we define the (narrow-sense, primitive) $q$-ary BCH code $\BCH\subseteq\F_q^{m'}$ as
\begin{equation}\label{eq:bchdef}
    \BCH= \left\{(c_0,\dots,c_{m'-1})\in\F_q: f(x)=\sum_{i=0}^{m'-1} c_i x^i\in\F_{q^r}[x] \textnormal{ satisfies } f(\alpha^i)=0 \textnormal{ for } i=1,\dots,d-1\right\}.
\end{equation}

By~\cite[Theorem 3]{Gur10}, it holds that $\BCH$ is a subset of a Reed-Solomon code over $\F_{q^r}$ with block length $m'$ and minimum distance $d$.
Therefore, $\BCH$ has minimum distance at least $d$.

To see the claim about the codimension, note that each constraint of the form $f(\alpha^i)=0$ over $\F_{q^r}$ corresponds to $r=\log_q(m'+1)$ linear constraints over $\F_q$.
Moreover, if $f(\gamma)=0$ for some $\gamma\in\F_{q^r}$, it follows that
\begin{equation*}
    0 = f(\gamma)^q = \sum_{i=0}^{m'-1} c_i^q \left(\gamma^{i}\right)^q = \sum_{i=0}^{m'-1} c_i \left(\gamma^{q}\right)^i = f(\gamma^q),
\end{equation*}
where we have used the fact that $(\beta_1+\beta_2)^q=\beta_1^q+\beta_2^q$ for any $\beta_1,\beta_2\in\F_{q^r}$ and that $\beta^q=\beta$ for all $\beta\in\F_q$ (with the natural embedding of $\F_q$ in $\F_{q^r}$).
Therefore, there are at least $\left\lfloor\frac{d-1}{q}\right\rfloor$ redundant constraints over $\F_{q^r}$ in \cref{eq:bchdef}.
Combining both observations above shows that the codimension of $\BCH$ is at most
\begin{equation*}
    \lceil (d-1)(1-1/q)\rceil \log_q(m'+1).
\end{equation*}

Finally, since each constraint $f(\alpha^i)=0$ over $\F_{q^r}$ can be transformed into $r$ linear constraints over $\F_q$ in time $\poly(q^r)=\poly(m')$, we can construct the parity-check matrix (and hence the generator matrix $\GBCH$) of $\BCH$ in time $\poly(m')$ as well.
\end{proof}
\section{The Haviv-Regev conditions for general $\ell_p$ norms}\label{sec:HRgen}

In this section, we prove a generalization of Haviv-Regev's conditions which allow tensoring of approximate $\SVP_p$ instances for all $p\geq 1$.
Our argument follows that of~\cite{HR12} for $p=2$ closely.
For convenience, we restate the key lemma here.
\havivregev*

We will require some auxiliary lemmas in order to prove \cref{lem:havivregev}, starting with a version of Minkowski's first theorem for (possibly non-full-rank) lattices in the $\ell_2$ norm.
\begin{lemma}[Minkowski's first theorem]\label{lem:minkowski}
    Let $\lat$ be a rank-$r$ lattice.
    Then, it holds that
    \begin{equation*}
        \det(\lat)\geq \left(\frac{\lambda_1^{(2)}(\lat)}{\sqrt{r}}\right)^r.
    \end{equation*}
\end{lemma}

We will also need to relate $\ell_p$ norms to the $\ell_2$ norm.
When $p\geq 2$, a standard application of H\"older's inequality yields
\begin{equation*}
    \|\vec{v}\|_2\geq \|\vec{v}\|_p\geq |\mathsf{supp}(\vec{v})|^{1/p-1/2} \|\vec{v}\|_2,
\end{equation*}
where $\mathsf{supp}(\vec{v})=\{i\in[m]:\vec{v}_i\neq 0\}$ is the support of $\vec{v}$, for any vector $\vec{v}\in\R^m$.
When $p<2$ we have that
\begin{equation*}
    \|\vec{v}\|_p\geq \|\vec{v}\|_2.
\end{equation*}
We can combine the inequalities above to conclude in particular that
\begin{equation}\label{eq:ellptoell2}
    \|\vec{v}\|_p\geq |\mathsf{supp}(\vec{v})|^{1/\max(2,p)-1/2} \|\vec{v}\|_2
\end{equation}
for all $p\geq 1$.
Moreover, using the fact that $\|\vec{v}\|_1 \leq |\mathsf{supp}(\vec{v})|^{1/2} \|\vec{v}\|_2$ by Cauchy-Schwarz and that $\|\vec{v}\|_1 \geq \|\vec{v}\|_p$ for all $p\geq 1$, it holds that
\begin{equation}\label{eq:ellptoell2alt}
    \|\vec{v}\|_2\geq |\mathsf{supp}(\vec{v})|^{-1/2} \|\vec{v}\|_p
\end{equation}
for all $p\geq 1$.

We prove a generalization of~\cite[Claim 3.5]{HR12} for $\ell_p$ norms with $p\geq 1$.
\begin{lemma}[\protect{Generalization of~\cite[Claim 3.5]{HR12}}]\label{lem:propsublattice}
    Let $(B,k)$ be a NO instance of $\gamma$-$\SVP_p$ with the properties outlined in \cref{lem:havivregev}, and let $\lat$ be a sublattice of $\lat(B)$ of rank $r$.
    Then, at least one of the following properties holds, where $d=\gamma k$:
    \begin{itemize}
        \item Every basis matrix of $\lat$ has more than $d^p$ nonzero rows;
        
        \item Every basis matrix of $\lat$ has only even entries and has more than $(d/2)^p$ nonzero rows;
        
        \item $\det(\lat)> d^{r(c+p/2)}$ and there is a basis matrix of $\lat$ which has at most $(d/2)^p$ nonzero rows.
        In particular, every vector in $\lat$ has Hamming weight at most $(d/2)^p$.
    \end{itemize}
\end{lemma}
\begin{proof}
    We assume that the first two properties do not hold and show that the third property holds in that case.
    Since there is a basis matrix of $\lat$ with at most $d^p$ nonzero rows, we conclude that $r\leq d^p$ and that every vector in $\lat$ has Hamming weight at most $d^p$.
    By the properties of $(B,k)$, this implies that $\lat\subseteq 2\Z^m$.
    Therefore, it must be the case that there is a basis matrix of $\lat$ that has at most $(d/2)^p$ nonzero rows, and so every vector in $\lat$ has Hamming weight at most $(d/2)^p$.
    As a result, we also gather that all nonzero vectors $\vec{w}\in\lat$ satisfy $\|\vec{w}\|_p> d^{c+3p/2}$, and so
    \begin{equation*}
        \det(\lat)\geq \left(\frac{\lambda_1^{(2)}(\lat)}{\sqrt{r}}\right)^r\geq \left(\frac{d^{-p/2}\lambda_1^{(p)}(\lat)}{\sqrt{r}}\right)^r> d^{r(c+p/2)} ,
    \end{equation*}
    where the first inequality follows from \cref{lem:minkowski}, the second inequality holds by \cref{eq:ellptoell2alt} and the fact that every vector in $\lat$ has Hamming weight at most $d^p$, and the third inequality uses the fact that $r\leq d^p$.
\end{proof}

We will also need the following technical lemma from~\cite{HR12}, which is specific for the $\ell_2$ norm.
\begin{lemma}[\protect{\cite[Claim 3.6]{HR12}}]\label{lem:tech}
    Let $\lat_1$ and $\lat_2$ be integer lattices of rank $r \geq 1$ generated by the bases
    $U=(\vec{u}_1,\dots,\vec{u}_r)$ and $W=(\vec{w}_1,\dots,\vec{w}_r)$, respectively.
    Consider the vector $\vec{v}=\sum_{i=1}^r \vec{u}_i\otimes \vec{w}_i\in\lat_1\otimes\lat_2$.
    Then, it holds that
    \begin{equation*}
        \|\vec{v}\|_2\geq \sqrt{r}(\det(\lat_1)\cdot \det(\lat_2))^{1/r}.
    \end{equation*}
\end{lemma}

We are now ready to prove \cref{lem:havivregev}.
\begin{proof}[Proof of \cref{lem:havivregev}]
    It suffices to show that if $(B,k)$ is a NO instance of $\gamma$-$\SVP_p$ satisfying the conditions from the lemma statement, $\lat_1=\lat(B)$, and $\lat_2=\bigotimes_{i=1}^{c'} \lat$ for some integer $1\leq c'<c$ and such that $\lambda_1^{(p)}(\lat_2)>d^{c'}$, then
    \begin{equation*}
        \lambda_1^{(p)}(\lat_1\otimes\lat_2)> d^{c'+1}.
    \end{equation*}
    
    Consider an arbitrary nonzero vector $\vec{v}\in\lat_1\otimes\lat_2$.
    As shown in~\cite[Proof of Lemma 3.4]{HR12}, 
    we can write $\vec{v}=B'_1 (B'_2)^T$ for 
    full-column-rank matrices $B'_1$ and $B'_2$ (note that $B'_1$ and $B'_2$ have the same number of columns)
    such that $\lat'_i=\lat(B'_i)\subseteq \lat_i$ for $i=1,2$.
    We now proceed by case analysis based on \cref{lem:propsublattice} applied to $\lat'_1$:
    \begin{itemize}
        \item $B'_1$ has more than $d^p$ nonzero rows:
        In this case, more than $d^p$ rows of $B'_1 (B'_2)^T$ are nonzero vectors from $\lat'_2$, and so
        \begin{equation*}
            \|\vec{v}\|_p > d\cdot \lambda_1^{(p)}(\lat'_2)\geq d \cdot \lambda_1^{(p)}(\lat_2)> d^{c'+1}.
        \end{equation*}
        
        \item $B'_1$ has only even entries and has more than $(d/2)^p$ nonzero rows:
        In this case, more than $(d/2)^p$ rows of $B'_1 (B'_2)^T$ are even multiples of nonzero vectors from $\lat'_2$, and so
        \begin{equation*}
            \|\vec{v}\|_p > 2((d/2)^p)^{1/p}\lambda_1^{(p)}(\lat'_2)\geq d \cdot \lambda_1^{(p)}(\lat_2)> d^{c'+1}.
        \end{equation*}
        
        \item $\det(\lat'_1)> d^{r(c+p/2)}$ and $B'_1$ has at most $(d/2)^p$ nonzero rows. In particular, $r\leq (d/2)^p$ and every vector in $\lat'_1$ has Hamming weight at most $(d/2)^p$:
        In this case, we have 
        \begin{align*}
            \|\vec{v}\|_p &\geq d^{\frac{p}{\max(2,p)}-\frac{p}{2}}\|\vec{v}\|_2\\
            &\geq d^{\frac{p}{\max(2,p)}-\frac{p}{2}} \sqrt{r}(\det(\lat'_1)\cdot\det(\lat'_2))^{1/r}\\
            &\geq d^{\frac{p}{\max(2,p)}-\frac{p}{2}} \sqrt{r}\det(\lat'_1)^{1/r}\\
            &> d^{\frac{p}{\max(2,p)}-\frac{p}{2}}\cdot d^{c+p/2}\\
            &\geq d^{c'+1} \ .
        \end{align*}
        The first inequality follows from \cref{eq:ellptoell2} and the fact that $\vec{v}$ has support size at most $d^p$.
        The second inequality holds via \cref{lem:tech} applied to $\lat'_1$ and $\lat'_2$, which are both rank-$r$ lattices.
        The third inequality uses the fact that $\det(\lat'_2)\geq 1$, which holds because $\lat'_2$ is a non-trivial integer lattice (this in turn holds by the definition of $\lat'_2$ and the fact that $\vec{v} \neq \vec{0}$.).
        The fourth inequality holds by the lower bound on $\det(\lat'_1)$ and the fact that $c'\leq c-1$.\qedhere
    \end{itemize}
\end{proof}

\end{document}